\documentclass[
  10pt,
  letterpaper,
  twocolumn,
  amsmath,
  amssymb,
  aps,
  pra,
  floatfix,
  superscriptaddress,
  longbibliography
]{revtex4-2}

\usepackage{iftex}
\ifPDFTeX
  \usepackage[utf8]{inputenc}
  \usepackage[T1]{fontenc}
\fi
\usepackage{graphicx}
\usepackage[caption=false,position=top]{subfig}
\usepackage{multirow}
\usepackage{dcolumn}

\usepackage{bm}
\usepackage{braket}
\usepackage{microtype}
\usepackage{hyperref}
\usepackage{amsthm}
\newtheorem{proposition}{Proposition}
\usepackage{todonotes}
\usepackage[ruled, vlined]{algorithm2e}
\usepackage{tikz}
\usetikzlibrary{quantikz2}
\theoremstyle{plain}
\newtheorem{theorem}{Theorem}
\newtheorem{lemma}[theorem]{Lemma}
\newtheorem{corollary}[theorem]{Corollary}

\theoremstyle{definition} 
\newtheorem{definition}{Definition}
\newtheorem{remark}{Remark}

\hypersetup{
  colorlinks=true,
  linkcolor=blue,
  citecolor=blue,
  urlcolor=blue,
  breaklinks=true
}

\begin{document}

\title{Dealing with locality in QAOA}

\thanks{These authors have contributed equally}

\author{Mithilesh Kumar}
\thanks{\href{mailto:mithilesh.kumar@krea.edu.in}{\nolinkurl{mithilesh.kumar@krea.edu.in}}}
\affiliation{Krea University, Sri City, India}
\author{Yusuf Tahir}
\thanks{\href{mailto:yusuf_tahir.sias25@krea.ac.in}{\nolinkurl{yusuf_tahir.sias25@krea.ac.in}}}
\affiliation{Krea University, Sri City, India}


\begin{abstract}
Shallow-depth QAOA on sparse, high-diameter MaxCut instances faces a locality bottleneck: at depth \(p\), local observables can depend only on a bounded neighborhood of the circuit interaction graph. We propose a transport-augmented QAOA that keeps the MaxCut cost Hamiltonian unchanged but enriches the mixer with optimized, unweighted shortcut couplings (scheduled \(XX+YY\)) to collapse the effective interaction-graph diameter. Using exact finite-depth support recursions, we relate optimal shortcut placement to bounded-diameter graph augmentation, and show in benchmarks that (unlike ma-QAOA) performance becomes effectively size-invariant once the diameter is reduced. For bipartite families (base diameter 4), reducing the interaction path to \(d=1\) raises the ensemble-averaged approximation ratio from 0.7378 (ma-QAOA) to 0.9767 at \(p=1\) (\(\sigma=0.0251\), nine system sizes); on random trees (base diameter 10), at \(p=2\) it improves from 0.9226 to 0.9997 (\(\sigma=0.0001\)).
\end{abstract}

\maketitle

\section{Introduction}

The Quantum Approximate Optimization Algorithm (QAOA), introduced by Farhi \emph{et al.}~\cite{farhi2014}, is among the most widely studied variational quantum algorithms for combinatorial optimization problems; see also Refs.~\cite{zhou2020,rendl2008solving,Hadfield2019}. For MaxCut on a graph \(G=(V,E)\), the standard depth-\(p\) QAOA state is
\[
|\psi_p(\boldsymbol{\gamma},\boldsymbol{\beta})\rangle
=
\prod_{\ell=1}^{p}
e^{-i\beta_\ell H_B}
e^{-i\gamma_\ell H_C}
|+\rangle^{\otimes n}
\]
where the cost Hamiltonian is
\[
H_C
=
\sum_{(i,j)\in E}
\frac12(1-Z_iZ_j)
\]
and the standard mixer Hamiltonian is
\[
H_B=\sum_i X_i
\]

A fundamental limitation of finite-depth QAOA is locality. After \(p\) layers, the expectation value of a local observable depends only on a bounded graph neighborhood of that observable, in the same spirit as finite-velocity propagation bounds in local quantum systems~\cite{lieb1972,hastings2006,bravyi2006}. Consequently, shallow circuits cannot make a local cost term depend on distant parts of a high-diameter graph. This locality phenomenon underlies several limitations of shallow QAOA on sparse graph families~\cite{bravyi2018,farhi2020}. In particular, on typical sparse instances, QAOA may need to ``see'' essentially the whole graph to reach near-optimal performance~\cite{farhi2020}.

A number of QAOA variants enrich the ansatz while preserving the underlying optimization objective, including generalized alternating-operator frameworks~\cite{Hadfield2019}, warm-start initializations~\cite{Egger2021}, and multi-angle parameterizations~\cite{Herrman2022}. Relatedly, Vijendran \emph{et al.} introduced an expressive low-depth ansatz (XQAOA) that extends mixer-augmented QAOA by upgrading the mixer from purely Pauli-$X$ rotations to an $XY$-type mixer (adding a Pauli-$Y$ component) without increasing the number of variational parameters, and observed strong performance and scaling behaviour in MaxCut benchmarks~\cite{Vijendran_2024}; they subsequently applied XQAOA-style mixer design ideas to the Binary Paint Shop Problem and reported strong scaling in large-instance simulations~\cite{Vijendran_2025_PaintShop}. Separately, several quantum platforms naturally provide interactions beyond nearest-neighbor connectivity, including long-range Ising couplings in trapped-ion systems, which can change the effective causal structure available to variational dynamics~\cite{Pagano2020}.

One way to address this limitation is to modify the graph perceived by the QAOA circuit. Phantom-QAOA, proposed by Langfitt \textit{et al.}~\cite{langfitt2024}, augments the problem Hamiltonian with additional weighted ``phantom'' edges, enabling the phase operator to encode a modified graph with increased visibility at a fixed circuit depth while preserving the original optimization objective. The method introduces only a single additional tunable parameter, $\alpha$, independent of the graph size and QAOA depth. The authors derive a general analytical expression for the Max-Cut expectation value at $p=1$ and show that, for even cycle graphs, Phantom-QAOA improves the approximation ratio from $0.75$ to $0.7925$, representing a $5.67\%$ relative improvement over standard QAOA. Moreover, numerical experiments on random regular graphs with up to 16 vertices demonstrate average approximation-ratio improvements of approximately $4\%$ at $p=1$ and $2\%$ at $p=2$. The study further finds that the proposed \textit{Triangle} edge-placement strategy consistently matches or outperforms the \textit{Full} edge-placement approach, indicating that the introduction of triangle-forming phantom edges is particularly effective for enhancing QAOA performance.

The present work studies a generic mechanism. We keep the MaxCut cost Hamiltonian exactly on the original instance graph, and instead modify the mixer interaction geometry. The added edges are not additional terms in the objective; they are transport channels used by the variational circuit. Thus the problem Hamiltonian remains
\[
H_C
=
\sum_{(i,j)\in E}\frac12(1-Z_iZ_j)
\]
while the mixer is allowed to contain two-qubit transport interactions on a chosen shortcut graph.

The main transport mixer considered in this paper is a scheduled \(XX+YY\) mixer. On a single edge,
\[
X_iX_j+Y_iY_j
=
2(\sigma_i^+\sigma_j^-+\sigma_i^-\sigma_j^+)
\]
so the interaction has the interpretation of excitation hopping between graph vertices. This makes it a natural transport mechanism and connects directly to native \(XY\)- and iSWAP-type hardware interactions. However, \(XX+YY\) terms on adjacent edges do not commute. To obtain exact compact lightcone statements, we therefore formulate the \(XX+YY\) ansatz as a scheduled edge-color circuit.

A commuting \(XX\) transport mixer is used as an analytically transparent intermediate model:
\[
H_B'
=
\sum_i X_i
+
\sum_{(i,j)\in E_{\mathrm{new}}}
X_iX_j
\]
where \(E_{\mathrm{new}}\subseteq\binom{V}{2}\) is a set of shortcut mixer edges. Because all \(X_iX_j\) terms commute, this intermediate model yields a particularly clean support recursion. The same graph-theoretic logic then extends to scheduled \(XX+YY\) transport, with the number of edge-color sublayers determining the exact lightcone velocity.

Throughout the paper we distinguish three graphs. The cost graph \(G_C=(V,E)\) defines the optimization Hamiltonian. The mixer graph \(G_M=(V,E_{\mathrm{new}})\) defines the added transport couplings. The full interaction graph \(G_{\mathrm{int}}=(V,E\cup E_{\mathrm{new}})\) controls the causal geometry of the circuit, because cost layers act on edges of \(E\) and mixer layers act on edges of \(E_{\mathrm{new}}\).

Our main results are finite-depth causal statements. For the commuting \(XX\) intermediate mixer, we prove that the Heisenberg support of an initially local operator is contained in the iterated growth region \(\Phi^p(S)\), where \(\Phi(S)=N_C(N_M(S))\). This implies the coarser radius bound \(\operatorname{supp}(U^\dagger O_SU)\subseteq N_{2p}^{\mathrm{int}}(S)\). For scheduled \(XX+YY\) transport with \(q\) matching sublayers per mixer layer, the exact growth map becomes \(\Psi_q(S)=N_C(N_M^q(S))\), and the corresponding radius bound is \((q+1)p\). We also derive exact commutator consequences, objective-locality statements for local MaxCut contributions, connected-correlation obstructions for product initial states, lightcone-volume estimates, and full-lightcone criteria.

This viewpoint leads naturally to a generic graph-augmentation problem: for a target depth, add as few shortcut mixer edges as possible so that the resulting interaction graph has sufficiently small diameter. We show that an optimal solution to this graph-theoretic problem is precisely the smallest mixer augmentation that removes the diameter-based causal obstruction identified by our support and commutator theorems. The numerical figures presented near the end of the paper illustrate the behavior of standard QAOA, augmented \(XX\), and augmented \(XX+YY\) mixers on random connected bipartite and random regular graph instances. These numerical results motivate the focus on \(XX+YY\) transport, while the rigorous results of the paper remain causal and graph-theoretic rather than claims of universal optimization improvement.

\section*{Contributions}
The main contributions of this paper are:
\begin{itemize}
  \item \textbf{Locality as interaction-graph geometry:} we formulate finite-depth QAOA limitations in terms of the circuit interaction graph \(G_{\mathrm{int}}=(V,E\cup E_{\mathrm{new}})\), keeping the MaxCut cost Hamiltonian fixed on the instance graph \(G=(V,E)\).
  \item \textbf{Transport-augmented mixers:} we introduce shortcut-transport mixers that add unweighted couplings only in the mixer, and analyze both a commuting \(XX\) intermediate model (for exact algebraic control) and a scheduled \(XX+YY\) transport mixer (hardware-motivated exchange/iSWAP dynamics).
  \item \textbf{Exact finite-depth lightcone recursions:} we derive exact support-growth maps for Heisenberg-evolved local observables (\(\Phi\), \(\Psi_q\)), yielding explicit radius bounds, full-lightcone criteria, and commutator-based obstructions to long-range influence at shallow depth.
  \item \textbf{Shortcut placement as diameter collapse:} we cast shortcut design as a bounded-diameter graph augmentation problem and show how optimal shortcut injection removes the diameter-based causal obstruction at depth \(p\).
  \item \textbf{Benchmarks and scaling:} using exact statevector simulations, we show a structural separation from ma-QAOA: ma-QAOA remains constrained by the native topology, while the transport ansatz becomes nearly size-invariant once the effective interaction diameter collapses. For bipartite families (base diameter 4), reducing the interaction path to \(d=1\) increases the ensemble-averaged approximation ratio from 0.7378 (ma-QAOA) to 0.9767 at \(p=1\) (\(\sigma=0.0251\), nine system sizes); for random trees (base diameter 10), at \(p=2\) it improves from 0.9226 to 0.9997 (\(\sigma=0.0001\)).
\end{itemize}

\section{Baseline QAOA and Finite-Depth Locality}

Combinatorial optimization objectives can be written as a sum of local clause contributions,
\begin{equation}
  C(z)=\sum_{\alpha=1}^{m} C_\alpha(z)
\end{equation}
where $z\in\{0,1\}^n$ and each $C_\alpha(z)$ depends only on a small subset of bits (e.g., an indicator for whether clause $\alpha$ is satisfied). In the quantum setting, $C$ is promoted to a diagonal operator in the computational basis and implemented via the cost-phase unitary
\begin{equation}
  U(C,\gamma)=e^{-i\gamma C}=\prod_{\alpha=1}^{m} e^{-i\gamma C_\alpha}
\end{equation}

For MaxCut on $G=(V,E)$ we take
\begin{equation}
  C=\sum_{(j,k)\in E} C_{jk},\qquad
  C_{jk}=\tfrac12(1-Z_jZ_k)
\end{equation}
The standard mixer Hamiltonian is
\begin{equation}
  B=\sum_{j\in V} X_j
\end{equation}
with mixer unitary $U(B,\beta)=e^{-i\beta B}$.

Starting from the uniform superposition state $\ket{s}=\ket{+}^{\otimes n}$, the depth-$p$ QAOA state with layer angles $\{(\gamma_\ell,\beta_\ell)\}_{\ell=1}^p$ is
\begin{equation}
  \ket{\gamma,\beta}
  =U(B,\beta_p)U(C,\gamma_p)\cdots U(B,\beta_1)U(C,\gamma_1)\ket{s}
\end{equation}
The objective value and optimal depth-$p$ value are
\begin{equation}
  F_p(\gamma,\beta)=\bra{\gamma,\beta}C\ket{\gamma,\beta}
  \qquad
  M_p=\max_{\gamma,\beta} F_p(\gamma,\beta)
\end{equation}

The crucial observation for finite-depth QAOA is that operator spreading remains local.

\subsection{Fixed-$p$ locality}

For MaxCut on a graph $G=(V,E)$ we define

\begin{equation}
    C
    =
    \sum_{(j,k)\in E} C_{jk}
\end{equation}

where

\begin{equation}
    C_{jk}
    =
    \frac12(1-Z_jZ_k)
\end{equation}

The cost expectation at depth $p$ is

\begin{equation}
F_p(\gamma,\beta)
=
\sum_{(j,k)\in E}
\bra{s}
U^\dagger(\gamma,\beta)
C_{jk}
U(\gamma,\beta)
\ket{s}
\end{equation}

Define

\begin{equation}
A_{jk}
=
U^\dagger(\gamma,\beta)
C_{jk}
U(\gamma,\beta)
\end{equation}

\begin{lemma}[Locality of standard QAOA \cite{farhi2014}]
The operator $A_{jk}$ only acts nontrivially on qubits $j,k$ and qubits whose graph distance from $j$ or $k$ is at most $p$.
\end{lemma}

\begin{proof}
Consider first the case $p=1$. Then

\begin{equation}
A_{jk}
=
U^\dagger(C,\gamma_1)
U^\dagger(B,\beta_1)
C_{jk}
U(B,\beta_1)
U(C,\gamma_1)
\end{equation}

All factors in $U(B,\beta_1)$ not acting on qubits $j$ or $k$ commute through $C_{jk}$ and cancel, leaving

\begin{equation}
A_{jk}
=
U^\dagger(C,\gamma_1)
e^{i\beta_1(X_j+X_k)}
C_{jk}
e^{-i\beta_1(X_j+X_k)}
U(C,\gamma_1)
\end{equation}

Next, all factors in $U(C,\gamma_1)$ not involving edges adjacent to $j$ or $k$ also commute through and cancel. Thus the support expands only to neighboring edges.

Repeating the argument layer-by-layer shows that after depth $p$, the support can spread at most $p$ graph hops from the original edge $(j,k)$.
\end{proof}

Consequently, for bounded-degree graphs, each contribution to $F_p$ depends only on a finite-radius neighborhood independent of the total system size $n$. This locality phenomenon underlies many known finite-depth limitations of QAOA on sparse graphs~\cite{bravyi2018,farhi2020}.

\section{\(XX+YY\) Transport-Augmented QAOA}

The preceding locality argument shows that finite-depth QAOA is fundamentally constrained by graph distance. A local cost term cannot depend on distant parts of the graph unless the circuit depth is sufficiently large compared to the graph diameter.
We now introduce a modified mixer designed to enlarge the interaction graph seen by the circuit.

We consider MaxCut on a graph $G=(V,E)$ (which determines the cost Hamiltonian), but allow the mixer interactions to act on a transport graph \(G_M=(V,E_{\mathrm{new}})\).

The key idea is to add shortcut transport edges so that the full interaction graph \(G_{\mathrm{int}}=(V,E\cup E_{\mathrm{new}})\) has smaller diameter than the original instance graph.

\begin{definition}[Commuting \(XX\) intermediate mixer family]
Let \(E_{\mathrm{new}} \subseteq \binom{V}{2}\) be a set of additional edges chosen to reduce the diameter of the full interaction graph

\begin{equation}
G_{\mathrm{int}}=(V,E\cup E_{\mathrm{new}})
\end{equation}

Define the mixer Hamiltonian
\begin{align}
H_B(\beta)
  &= \beta_0\sum_{j\in V}X_j
     + \sum_{(i,j)\in E_{\mathrm{new}}}\beta_{ij}X_iX_j
\end{align}

The corresponding mixer unitary is
\begin{align}
U_B(\beta) &= e^{-iH_B(\beta)}
\end{align}

The depth-$p$ transport-augmented QAOA state is
\begin{align}
\ket{\gamma,\beta}
  &= \prod_{\ell=1}^{p} U_B(\beta^{(\ell)})\,U(C,\gamma_\ell)\,\ket{s}
\end{align}

\end{definition}

\begin{definition}[Scheduled \(XX+YY\) transport mixer]
An \(XX+YY\) transport mixer on the same transport graph \(G_M=(V,E_{\mathrm{new}})\) is implemented as a finite sequence of two-qubit gates supported on the transport edges. Let
\[
E_{\mathrm{new}}=\mathcal M_1\cup\cdots\cup \mathcal M_q
\]
be an edge coloring of \(G_M\), so each \(\mathcal M_c\) is a matching. For each color class define
\[
V_c(\theta)
:=
\prod_{(i,j)\in \mathcal M_c}
\exp\!\left[
-i\left(
\theta_{ij}^{x}X_iX_j+\theta_{ij}^{y}Y_iY_j
\right)
\right]
\]
The scheduled \(XX+YY\) mixer layer is
\[
U_B^{XY}(\theta)
:=
\left(
\prod_{c=1}^{q}V_c(\theta)
\right)
\left(
\prod_{u\in V}e^{-i\eta_u X_u}
\right)
\]
with a fixed ordering of the color classes. The single-qubit rotations are optional and do not enlarge support.

\end{definition}

\subsection{Circuit implementation of \(XX\) and \(XX+YY\) transport mixers}

Each additional transport term \(e^{-i\beta_{ij}X_iX_j}\) is a two-qubit Pauli rotation.

Using Hadamard conjugation,

\begin{equation}
e^{-i\beta_{ij}X_iX_j}
=
(H_i\otimes H_j)
e^{-i\beta_{ij}Z_iZ_j}
(H_i\otimes H_j)
\end{equation}

The central $ZZ$ rotation admits the standard decomposition

\begin{equation}
e^{-i\beta Z_iZ_j}
=
\mathrm{CNOT}_{ij}
\,R_z(2\beta)_j\,
\mathrm{CNOT}_{ij}
\end{equation}
Here we use the convention \(R_z(\phi)=e^{-i\phi Z/2}\).

Therefore each transport edge introduces a two-CNOT entangling block:

\begin{equation}
\begin{aligned}
e^{-i\beta_{ij}X_iX_j}
&=
(H_i\otimes H_j)
\,
\mathrm{CNOT}_{ij}
\,R_z(2\beta_{ij})_j\,
\mathrm{CNOT}_{ij}
\\
&\quad\times
(H_i\otimes H_j)
\end{aligned}
\end{equation}

Equivalently, defining
\[
\begin{aligned}
R_{XX}^{ij}(\theta)&:=e^{-i\theta X_iX_j}\\
R_{YY}^{ij}(\theta)&:=e^{-i\theta Y_iY_j}\\
R_{ZZ}^{ij}(\theta)&:=e^{-i\theta Z_iZ_j}
\end{aligned}
\]
we have
\[
R_{XX}^{ij}(\theta)
=
(H_i\otimes H_j)
R_{ZZ}^{ij}(\theta)
(H_i\otimes H_j)
\]
Thus an \(XX\) mixer edge requires two CNOTs and single-qubit Hadamards when implemented through a \(ZZ\) rotation.

Hence the transport-augmented mixer introduces explicit entangling transport channels between distant regions of the graph.

Because all $X_iX_j$ terms commute with one another and with all single-qubit $X_i$ terms,

\begin{equation}
[X_iX_j, X_kX_l]=0
\qquad
[X_iX_j,X_k]=0
\end{equation}

the entire mixer layer can be parallelized according to an edge-coloring of the transport graph $G_M$.

For the \(XX+YY\) variant, each edge gate
\[
\exp\!\left[-i\left(\theta_{ij}^{x}X_iX_j+\theta_{ij}^{y}Y_iY_j\right)\right]
\]
is also a two-qubit gate. On a single edge, the two Pauli strings commute:
\[
[X_iX_j,Y_iY_j]=0
\]
Thus the edge gate can be written as
\[
e^{-i\theta_{ij}^{x}X_iX_j}
e^{-i\theta_{ij}^{y}Y_iY_j}
\]
The \(YY\) rotation is reduced to the same \(ZZ\)-rotation block by the phase-basis change \(S=\operatorname{diag}(1,i)\):
\[
Y=S H Z H S^\dagger
\]
Therefore
\begin{equation}
e^{-i\theta Y_iY_j}
=
(S_iH_i\otimes S_jH_j)
e^{-i\theta Z_iZ_j}
(H_iS_i^\dagger\otimes H_jS_j^\dagger)
\end{equation}
and hence
\begin{equation}
\begin{aligned}
e^{-i\theta Y_iY_j}
&=
(S_iH_i\otimes S_jH_j)
\,
\mathrm{CNOT}_{ij}
\,R_z(2\theta)_j\,
\mathrm{CNOT}_{ij}
\\
&\quad\times
(H_iS_i^\dagger\otimes H_jS_j^\dagger)
\end{aligned}
\end{equation}
Consequently a general scheduled edge gate
\[
R_{XY}^{ij}(\theta_x,\theta_y)
:=
\exp[-i(\theta_xX_iX_j+\theta_yY_iY_j)]
\]
has the exact decomposition
\[
R_{XY}^{ij}(\theta_x,\theta_y)
=
R_{XX}^{ij}(\theta_x)R_{YY}^{ij}(\theta_y)
\]
Because this factors into two commuting Pauli rotations, a direct compilation of the product uses two CNOTs per factor, i.e., four CNOTs total.

However, the combined operator \(\exp[-i(\theta_xX_iX_j+\theta_yY_iY_j)]\) is a parity-preserving two-qubit gate (a matchgate) and can be synthesized with only two CNOTs, up to single-qubit gates \cite{vatan2004}. One convenient two-CNOT form is
\begin{eqnarray*}
\exp\!\left[-i \left(\theta_xX_iX_j+\theta_yY_iY_j\right)\right]
=\\
(K_i\otimes K_j)
\,\mathrm{CNOT}_{ij}\,
\bigl(R_y(2\theta_x)_i\otimes R_y(2\theta_y)_j\bigr)
\,\\ 
\mathrm{CNOT}_{ij}\,
(K_i^\dagger\otimes K_j^\dagger)
\end{eqnarray*}

where
\[
K_i=K_j=HS^\dagger
\qquad
(\text{so }K_i^\dagger=K_j^\dagger=SH)
\]
where \(H\) is the Hadamard and \(S=\operatorname{diag}(1,i)\).

However, \(XX+YY\) terms on adjacent transport edges need not commute. For distinct vertices \(a,b,c\),
\[
[X_aX_b+Y_aY_b,\;X_bX_c+Y_bY_c]
=
2i(X_aZ_bY_c-Y_aZ_bX_c)
\]
which is generally nonzero. Therefore the exact finite-depth support proof for \(XX+YY\) mixers uses the scheduled edge-color implementation above, where each matching sublayer consists of disjoint two-qubit gates.

\section{Variational Monotonicity Under Mixer Augmentation}

We first establish that adding additional mixer edges cannot reduce the best achievable variational optimum at fixed depth.

\begin{proposition}[Monotonic improvement under mixer-edge augmentation]

Fix the instance graph \(G=(V,E)\) and depth \(p\).
Assume the allowed mixer-parameter domains contain the origin.
Let \(\mathcal{F}_p(E_{\mathrm{new}})\) be the variational family above with additional mixer edges \(E_{\mathrm{new}}\).
Then for any \(E_{\mathrm{new}}' \supseteq E_{\mathrm{new}}\) we have
\begin{equation}
    \max_{\ket{\psi}\in \mathcal{F}_p(E_{\mathrm{new}}')}
    \bra{\psi}C\ket{\psi}
    \ge
    \max_{\ket{\psi}\in \mathcal{F}_p(E_{\mathrm{new}})}
    \bra{\psi}C\ket{\psi}
\end{equation}

In particular, adding diameter-reducing transport mixer edges cannot worsen the best achievable approximation ratio at fixed depth \(p\).

\end{proposition}

\begin{proof}

The family \(\mathcal{F}_p(E_{\mathrm{new}})\) is obtained from the larger family \(\mathcal{F}_p(E_{\mathrm{new}}')\) by restricting parameters: set \(\beta_{ij}^{(\ell)}=0\) for all added edges \((i,j)\in E_{\mathrm{new}}'\setminus E_{\mathrm{new}}\) and all layers \(\ell\).

Thus \(\mathcal{F}_p(E_{\mathrm{new}})\subseteq \mathcal{F}_p(E_{\mathrm{new}}')\),
and maximizing a fixed observable \(C\) over a superset cannot decrease the optimum.

\end{proof}

Locality in QAOA at fixed \(p\) is governed by the lightcone induced by the interaction graphs of both the cost layers and the mixer layers.

Adding transport terms on edges of \(E_{\mathrm{new}}\) changes the full interaction graph \(G_{\mathrm{int}}=(V,E\cup E_{\mathrm{new}})\), so conjugating a local cost term \(C_{jk}\) through alternating layers is upper bounded by neighborhoods in \(G_{\mathrm{int}}\).

Heuristically, a smaller interaction-graph diameter enlarges the portion of the instance graph that can influence a local expectation value at depth \(p\), which is the graph-distance obstruction behind locality-based lower bounds.

\section{Lightcone Theory for Transport-Augmented QAOA}

For an edge \((u,v)\in E\) define the local MaxCut observable \(C_{uv}=\frac12(1-Z_uZ_v)\).
The augmented QAOA unitary at depth \(p\) is \(U_p = \prod_{\ell=1}^{p} U_B(\beta^{(\ell)}) U(C,\gamma_\ell)\).
We study the Heisenberg evolution \(C_{uv}^{(p)} = U_p^\dagger C_{uv} U_p\).

\begin{definition}[\(r\)-neighborhood in the interaction graph]

Let \(G_{\mathrm{int}}=(V,E\cup E_{\mathrm{new}})\). For a subset \(S\subseteq V\) and integer \(r\ge0\), define
\begin{equation}
    N_r^{\mathrm{int}}(S)
    :=
    \{
    v\in V:
    \mathrm{dist}_{G_{\mathrm{int}}}(v,S)\le r
    \}
\end{equation}

For a single vertex \(u\in V\), we write \(B_{G_{\mathrm{int}}}(u,r):=N_r^{\mathrm{int}}(\{u\})\). In particular, for any two vertices \(u,v\in V\), \(N_r^{\mathrm{int}}(\{u,v\}) = B_{G_{\mathrm{int}}}(u,r)\cup B_{G_{\mathrm{int}}}(v,r)\).

\end{definition}

\begin{definition}[One-step growth maps]
For a subset \(S\subseteq V\), define the mixer expansion \(N_M(S):=S\cup\{v\in V:\exists\,u\in S\text{ with }(u,v)\in E_{\mathrm{new}}\}\), the cost expansion \(N_C(S):=S\cup\{v\in V:\exists\,u\in S\text{ with }(u,v)\in E\}\), and the full one-layer growth map \(\Phi(S):=N_C(N_M(S))\). Iterating, we write \(\Phi^0(S):=S\) and \(\Phi^{t+1}(S):=\Phi(\Phi^t(S))\).

\end{definition}

\begin{theorem}[Operator-growth recursion and lightcone bound]
Fix depth \(p\) and angles, and let
\(U_p := \prod_{\ell=1}^{p} U_B(\beta^{(\ell)}) U(C,\gamma_\ell)\)
be the corresponding depth-\(p\) alternating unitary.

For any operator \(O_S\) initially supported on a subset \(S\subseteq V\) define
\begin{equation}
  O_S^{(p)} := U_p^\dagger O_S U_p
\end{equation}
Then \(\operatorname{supp}(O_S^{(p)})\) is contained in the iterated growth region \(\Phi^p(S)\), and therefore also obeys the interaction-graph lightcone bound \(\operatorname{supp}(O_S^{(p)}) \subseteq N_{2p}^{\mathrm{int}}(S)\).
\end{theorem}

\begin{proof}

For \(t=0,1,\dots,p\), define \(U_t := \prod_{\ell=1}^{t} U_B(\beta^{(\ell)})U(C,\gamma_\ell)\) and \(O_S^{(t)} := U_t^\dagger O_S U_t\), with \(U_0=\mathbb{I}\), and let \(S_t:=\Phi^t(S)\). We first prove by induction on \(t\) that \(\operatorname{supp}(O_S^{(t)})\subseteq S_t\).

The base case \(t=0\) is immediate because \(O_S\) is supported on \(S=S_0\).

For the inductive step, assume \(\operatorname{supp}(O_S^{(t)})\subseteq S_t\).
Write the next mixer layer as
\begin{equation}
    U_B(\beta^{(t+1)})
    =
    \prod_{u\in V}
    e^{-i\beta_0^{(t+1)}X_u}
    \prod_{(a,b)\in E_{\mathrm{new}}}
    e^{-i\beta_{ab}^{(t+1)}X_aX_b}
\end{equation}
and the next cost layer as
\begin{equation}
    U(C,\gamma_{t+1})
    =
    \prod_{(a,b)\in E}
    e^{-i\gamma_{t+1}C_{ab}}
\end{equation}

If a unitary \(U_T\) acts trivially outside a subset \(T\subseteq V\), then \(\operatorname{supp}(U_T^\dagger O_R U_T)\subseteq R\cup T\) for every operator \(O_R\) supported on \(R\). The single-qubit \(X_u\) factors do not enlarge support. For the two-qubit \(XX\) factors, all mixer terms commute with one another. Hence all edge factors whose support is disjoint from \(S_t\) commute through \(O_S^{(t)}\) and cancel against their adjoints. The only remaining two-qubit mixer factors are supported on transport edges incident on \(S_t\), and each such factor can only adjoin the other endpoint of that edge. Therefore
\(\operatorname{supp}\!\left(U_B(\beta^{(t+1)})^\dagger O_S^{(t)} U_B(\beta^{(t+1)})\right)\subseteq N_M(S_t)\).

Conjugating next by \(U(C,\gamma_{t+1})\) can enlarge support by at most one additional hop along a cost edge \((a,b)\in E\subseteq E\cup E_{\mathrm{new}}\). Therefore \(\operatorname{supp}(O_S^{(t+1)})\subseteq N_C(N_M(S_t))=S_{t+1}\).
Thus \(\operatorname{supp}(O_S^{(t)})\subseteq S_t\) for all \(t\), and in particular \(\operatorname{supp}(O_S^{(p)})\subseteq \Phi^p(S)\).

To compare \(\Phi^t(S)\) with interaction-graph neighborhoods, note that \(N_M(S)\subseteq N_1^{\mathrm{int}}(S)\) and \(N_C(S)\subseteq N_1^{\mathrm{int}}(S)\),
because \(E_{\mathrm{new}}\subseteq E\cup E_{\mathrm{new}}\) and \(E\subseteq E\cup E_{\mathrm{new}}\). Therefore \(\Phi(S)=N_C(N_M(S))\subseteq N_2^{\mathrm{int}}(S)\). Inducting on \(t\) then gives \(\Phi^t(S)\subseteq N_{2t}^{\mathrm{int}}(S)\).
Combining this with the first part proves \(\operatorname{supp}(O_S^{(p)})\subseteq \Phi^p(S)\subseteq N_{2p}^{\mathrm{int}}(S)\).

\end{proof}

\begin{theorem}[Support bound for scheduled \(XX+YY\) transport mixers]
Let \(E_{\mathrm{new}}=\mathcal M_1\cup\cdots\cup\mathcal M_q\) be an edge coloring of the transport graph \(G_M\), so each \(\mathcal M_c\) is a matching. Let each mixer layer be the scheduled \(XX+YY\) unitary
\[
\begin{aligned}
U_B^{XY}(\theta^{(\ell)})
&=
\prod_{c=1}^{q}V_c(\theta^{(\ell)})
\prod_{u\in V}e^{-i\eta_u^{(\ell)}X_u},
\\
V_c(\theta^{(\ell)})&:=
\prod_{(a,b)\in \mathcal M_c}
\exp\!\left[-i\left(\theta_{ab}^{x,(\ell)}X_aX_b+\theta_{ab}^{y,(\ell)}Y_aY_b\right)\right]
\end{aligned}
\]
Define \(\Psi_q(S):=N_C(N_M^q(S))\), where \(N_M^0(S):=S\) and \(N_M^{m+1}(S):=N_M(N_M^m(S))\). For the depth-\(p\) unitary
\[
U_p^{XY}:=
\prod_{\ell=1}^{p}
U_B^{XY}(\theta^{(\ell)})U(C,\gamma_\ell)
\]
any operator \(O_S\) initially supported on \(S\subseteq V\) satisfies
\[
\operatorname{supp}\!\left((U_p^{XY})^\dagger O_S U_p^{XY}\right)
\subseteq
\Psi_q^p(S)
\subseteq
N_{(q+1)p}^{\mathrm{int}}(S)
\]
In particular, if \(q=1\), meaning the transport graph is a matching in each mixer layer, the same radius-\(2p\) bound as the commuting \(XX\) mixer is recovered.
\end{theorem}

\begin{proof}
We first prove the one-layer support update. Suppose an operator is supported on a set \(R\). A color class \(\mathcal M_c\) is a matching, so the gates in \(V_c\) have disjoint supports. Gates on edges disjoint from the current support commute through the operator and cancel against their adjoints. Gates on edges incident on the current support can only add the opposite endpoint of that edge. Therefore conjugation by one matching sublayer maps support into \(N_M(R)\).
Applying the \(q\) color classes in the scheduled order gives support contained in \(N_M^q(R)\). The single-qubit rotations \(e^{-i\eta_u X_u}\) do not enlarge support. Conjugating afterward by the cost layer \(U(C,\gamma_\ell)\) can add at most one cost-graph hop, so one full \(XX+YY\) QAOA layer maps \(R\) into \(N_C(N_M^q(R))=\Psi_q(R)\).

Inducting over \(p\) alternating layers gives
\[
\operatorname{supp}\!\left((U_p^{XY})^\dagger O_S U_p^{XY}\right)
\subseteq
\Psi_q^p(S)
\]
Finally, \(N_M(R)\subseteq N_1^{\mathrm{int}}(R)\) and \(N_C(R)\subseteq N_1^{\mathrm{int}}(R)\), so \(\Psi_q(R)=N_C(N_M^q(R))\subseteq N_{q+1}^{\mathrm{int}}(R)\). Iterating this inclusion yields \(\Psi_q^p(S)\subseteq N_{(q+1)p}^{\mathrm{int}}(S)\), which proves the theorem.

\end{proof}

For a commuting \(XX\) mixer, the Hamiltonian exponential
\[
\exp\!\left[-i\sum_{(i,j)\in E_{\mathrm{new}}}\beta_{ij}X_iX_j\right]
\]
has an exact one-hop support bound because all transport terms commute. By contrast, the analog \(XX+YY\) Hamiltonian
\[
\sum_{(i,j)\in E_{\mathrm{new}}}
\left(
\theta_{ij}^{x}X_iX_j+\theta_{ij}^{y}Y_iY_j
\right)
\]
contains noncommuting terms whenever transport edges share a vertex. 

The commuting \(XX\) mixer is algebraically convenient because all transport terms \(X_iX_j\) commute, so the entire mixer exponential has an exact one-mixer-hop support bound. The \(XX+YY\) mixer has a more direct transport interpretation:
\[
X_iX_j+Y_iY_j
=
2(\sigma_i^+\sigma_j^-+\sigma_i^-\sigma_j^+)
\]
so it moves a single excitation between \(i\) and \(j\) and, without the optional single-qubit \(X\) rotations, preserves Hamming weight. Its adjacent-edge terms do not commute, however, so an exact compact lightcone is obtained for the scheduled edge-color circuit rather than for the unscheduled Hamiltonian exponential; the price is that the lightcone velocity depends on the number \(q\) of matching sublayers.

\begin{corollary}[Full-lightcone criterion for scheduled \(XX+YY\)]
For a scheduled \(XX+YY\) mixer implemented with \(q\) edge-color sublayers, if
\[
\operatorname{diam}(G_{\mathrm{int}})\le (q+1)p
\]
then for every nonempty \(S\subseteq V\),
\[
N_{(q+1)p}^{\mathrm{int}}(S)=V
\]
Thus the graph-distance support bound for the scheduled \(XX+YY\) ansatz becomes the entire vertex set at depth \(p\). Equivalently, in the generic augmentation problem one replaces the commuting-\(XX\) target \(D=2p\) by
\[
D=(q+1)p
\]

\end{corollary}

\begin{proof}
The proof is identical to the full-lightcone criterion for the commuting \(XX\) mixer, with radius \(2p\) replaced by \((q+1)p\). If \(s\in S\) and \(v\in V\), then
\[
\mathrm{dist}_{G_{\mathrm{int}}}(v,S)
\le
\mathrm{dist}_{G_{\mathrm{int}}}(v,s)
\le
\operatorname{diam}(G_{\mathrm{int}})
\le
(q+1)p
\]
Therefore \(v\in N_{(q+1)p}^{\mathrm{int}}(S)\), and since \(v\) was arbitrary,
\[
N_{(q+1)p}^{\mathrm{int}}(S)=V
\]

\end{proof}

\begin{corollary}[Exact finite-depth commutator bound]
Let \(O_S\) and \(O_T\) be operators initially supported on subsets \(S,T\subseteq V\). If \(T\cap \Phi^p(S)=\emptyset\), then \([U_p^\dagger O_S U_p,\; O_T]=0\). In particular, if \(\mathrm{dist}_{G_{\mathrm{int}}}(S,T)>2p\), then \([U_p^\dagger O_S U_p,\; O_T]=0\).
\end{corollary}

\begin{proof}
By the theorem, \(\operatorname{supp}(U_p^\dagger O_S U_p)\subseteq \Phi^p(S)\). Hence \(U_p^\dagger O_S U_p\) and \(O_T\) act on disjoint qubit sets when \(T\cap\Phi^p(S)=\emptyset\), so they commute. The distance-based statement follows from \(\Phi^p(S)\subseteq N_{2p}^{\mathrm{int}}(S)\).
\end{proof}

\begin{remark}[Heisenberg lightcones versus amplitudes]
 The lightcone statements in this section are Heisenberg-picture statements about the support of evolved observables. They imply that the expectation value of a local observable, such as
 \[
 \langle \psi_p|C_{uv}|\psi_p\rangle
 =
 \langle s|U_p^\dagger C_{uv}U_p|s\rangle
 \]
 depends only on the qubits inside the corresponding evolved-operator support. They do not assert that individual computational-basis amplitudes in an expansion
 \[
 \ket{\psi_p}=\sum_z c_z\ket{z}
 \]
 factorize or depend only on local graph neighborhoods. A single amplitude \(c_z\) is a global Schrödinger-picture quantity involving many computational-basis paths. The locality result proved here is therefore the standard QAOA locality statement: local expectation values are controlled by the Heisenberg lightcone of the local observable.

 \end{remark}

\begin{corollary}[Full-lightcone criterion]
If
\[
\operatorname{diam}(G_{\mathrm{int}})\le 2p
\]
then for every nonempty subset \(S\subseteq V\)
\[
N_{2p}^{\mathrm{int}}(S)=V
\]
Consequently, for every operator \(O_S\) initially supported on \(S\), the graph-distance support bound from the theorem becomes the entire vertex set \(V\).

\end{corollary}

\begin{proof}
Fix a nonempty \(S\subseteq V\), choose \(s\in S\), and let \(v\in V\) be arbitrary. Then
\[
\mathrm{dist}_{G_{\mathrm{int}}}(v,S)
\le
\mathrm{dist}_{G_{\mathrm{int}}}(v,s)
\le
\operatorname{diam}(G_{\mathrm{int}})
\le
2p
\]
Hence \(v\in N_{2p}^{\mathrm{int}}(S)\), proving
\[
N_{2p}^{\mathrm{int}}(S)=V
\]

\end{proof}

\begin{remark}[Lieb--Robinson interpretation]

The theorem and exact commutator corollary constitute a finite-depth Lieb--Robinson statement specialized to an alternating local quantum circuit.

Operator support can expand by at most two graph hops per full QAOA layer: one arising from the mixer graph \(G_M\) and one from the cost graph \(G_C\), both of which are contained in \(G_{\mathrm{int}}\).

Thus reducing
\(
\mathrm{diam}(G_{\mathrm{int}})
\)
increases the portion of the graph that lies inside
\(
N_{2p}^{\mathrm{int}}(\{j,k\})\)
at fixed depth \(p\).

\end{remark}

\section{Consequences, Graph Augmentation, and Symmetry}

The operator-growth theorem above implies that finite-depth support propagation is controlled by distances in the interaction graph \(G_{\mathrm{int}}\), rather than the original graph \(G\) alone.
In standard QAOA, the Heisenberg support of a local edge observable can expand by at most one cost-graph hop per layer, so after depth \(p\) it is confined to a radius-\(p\) neighborhood in the original graph \(G\).
Under transport augmentation with a commuting \(XX\) transport mixer, the support can expand by at most one mixer hop and one cost hop per layer, so after depth \(p\) it is confined to a radius-\(2p\) neighborhood in the interaction graph \(G_{\mathrm{int}}\). More generally, for scheduled \(XX+YY\) transport with \(q\) matching sublayers per mixer layer, the radius bound becomes \((q+1)p\).
Therefore, shrinking the interaction-graph diameter (i.e., achieving \(\operatorname{diam}(G_{\mathrm{int}})<\operatorname{diam}(G)\)) relaxes finite-depth graph-distance obstructions by bringing more of the instance within the relevant lightcones at fixed depth \(p\).

\subsection{Generic Graph-Augmentation Framework}

We now make the graph-agnostic augmentation problem explicit. The cost Hamiltonian is always kept on the original instance graph \(G=(V,E)\), and only the mixer interaction pattern is enlarged.

\begin{definition}[Diameter-reducing augmentation problem]
Let \(G=(V,E)\) be a finite instance graph, and fix a target diameter \(D\ge1\). Define the feasible augmentation family
\[
\mathcal A_D(G):=\left\{
F\subseteq \binom{V}{2}\setminus E:
\operatorname{diam}(V,E\cup F)\le D
\right\}
\]
The corresponding diameter-reduction number is
\[
\alpha_D(G)
:=
\min_{F\in\mathcal A_D(G)} |F|
\]
Any minimizer \(F_D^\star\in\mathcal A_D(G)\) is called an optimal diameter-\(D\) augmentation.

Because \(V\) is finite and \(F=\binom{V}{2}\setminus E\) produces the complete graph \(K_{|V|}\) of diameter \(1\), the feasible set \(\mathcal A_D(G)\) is nonempty for every \(D\ge1\), and the minimum exists.

\end{definition}

\begin{figure*}[!tbp]
\centering
\subfloat[\label{fig:reduction-base}] {
\includegraphics[width=0.23\textwidth]{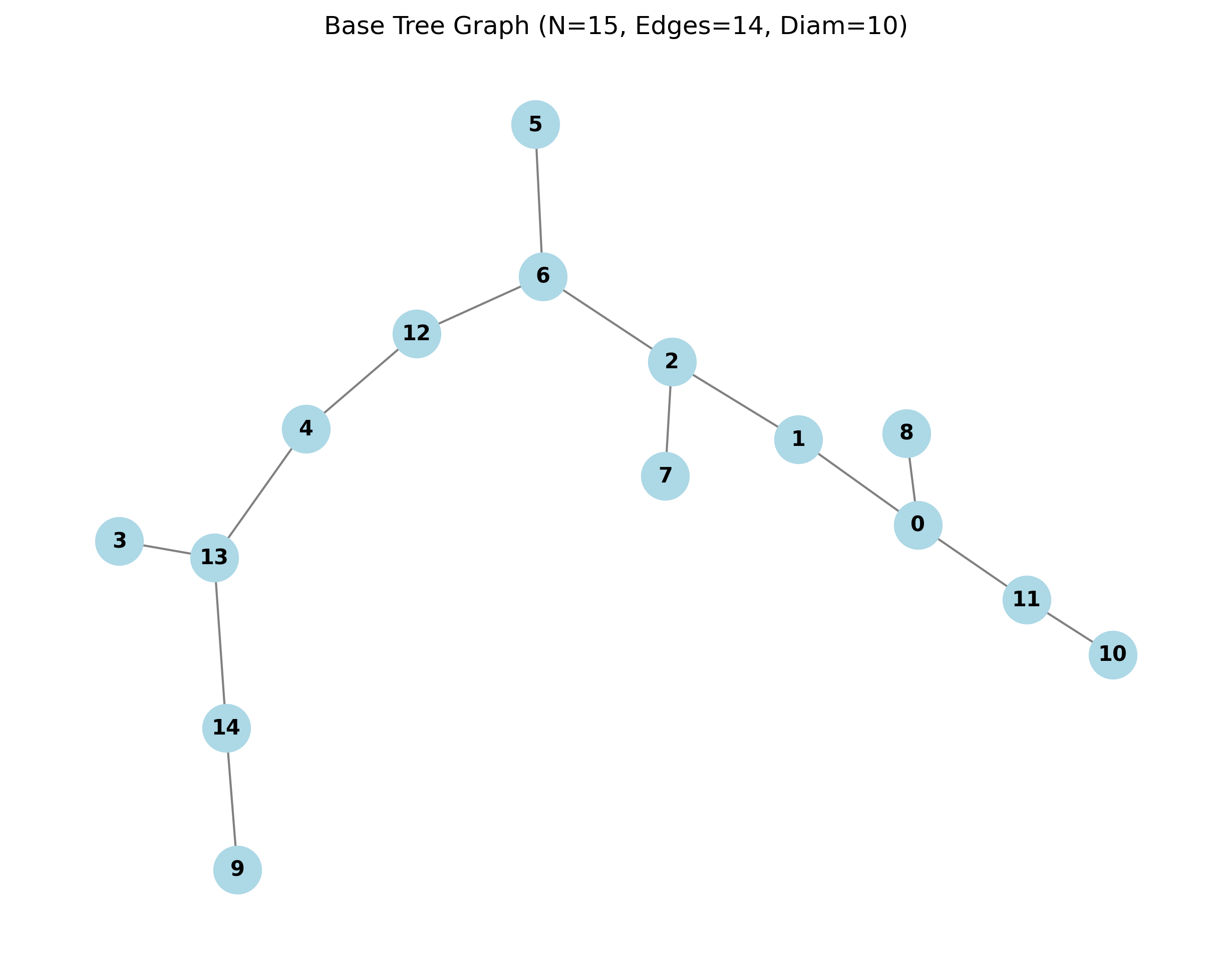}
}\hfill
\subfloat[\label{fig:reduction-d7}] {
\includegraphics[width=0.23\textwidth]{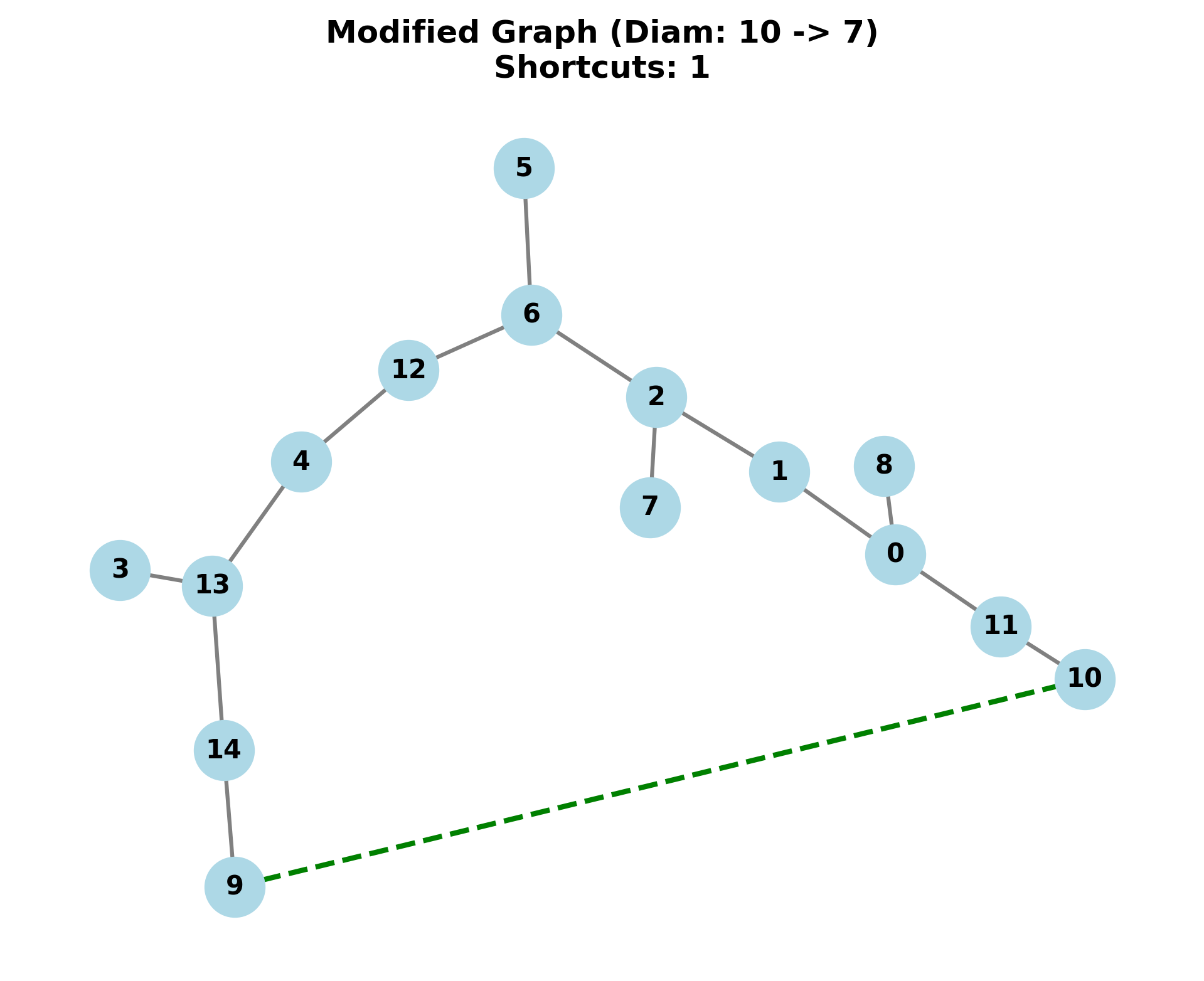}
}\hfill
\subfloat[\label{fig:reduction-d6}] {
\includegraphics[width=0.23\textwidth]{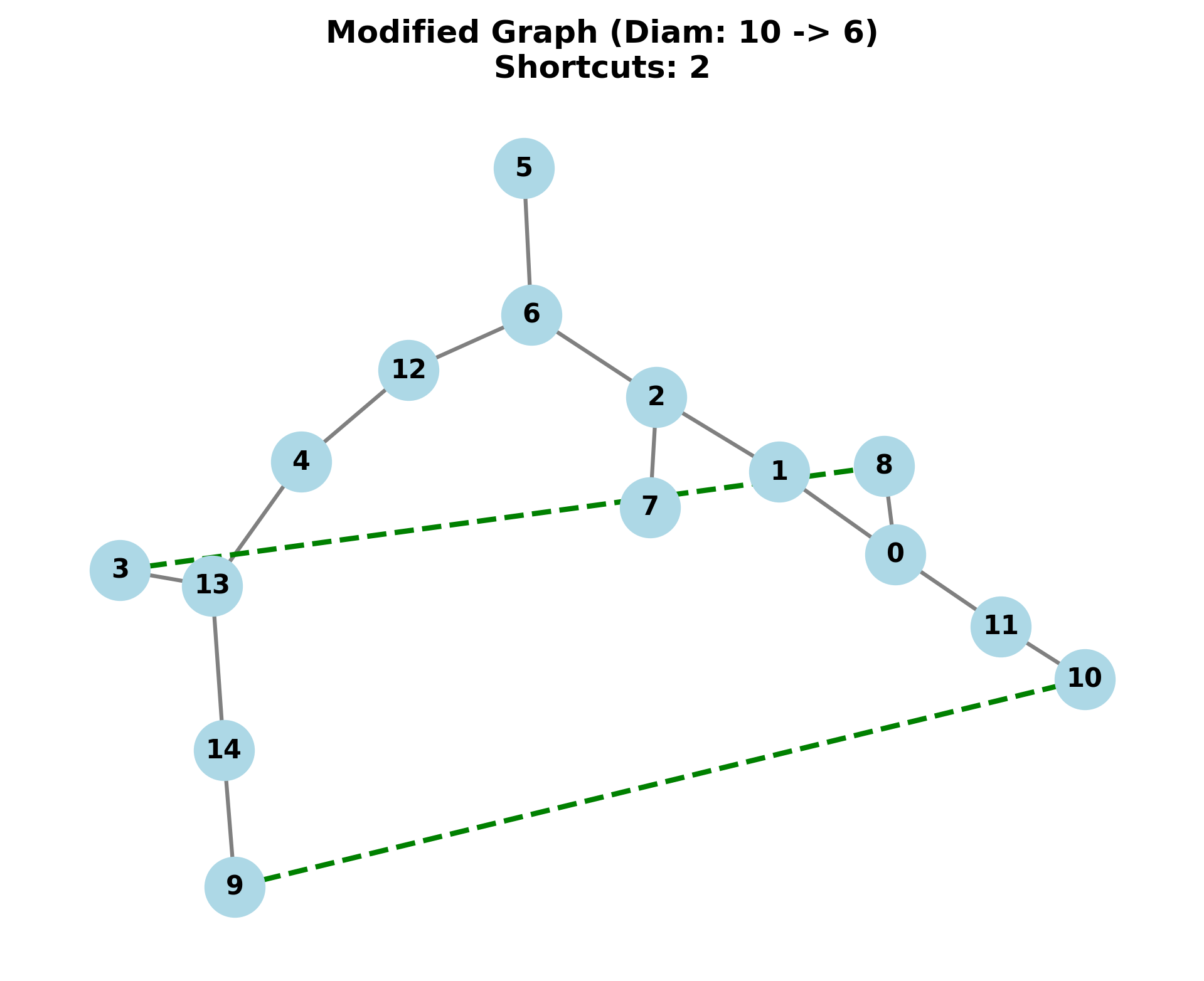}
}\hfill
\subfloat[\label{fig:reduction-d5}] {
\includegraphics[width=0.23\textwidth]{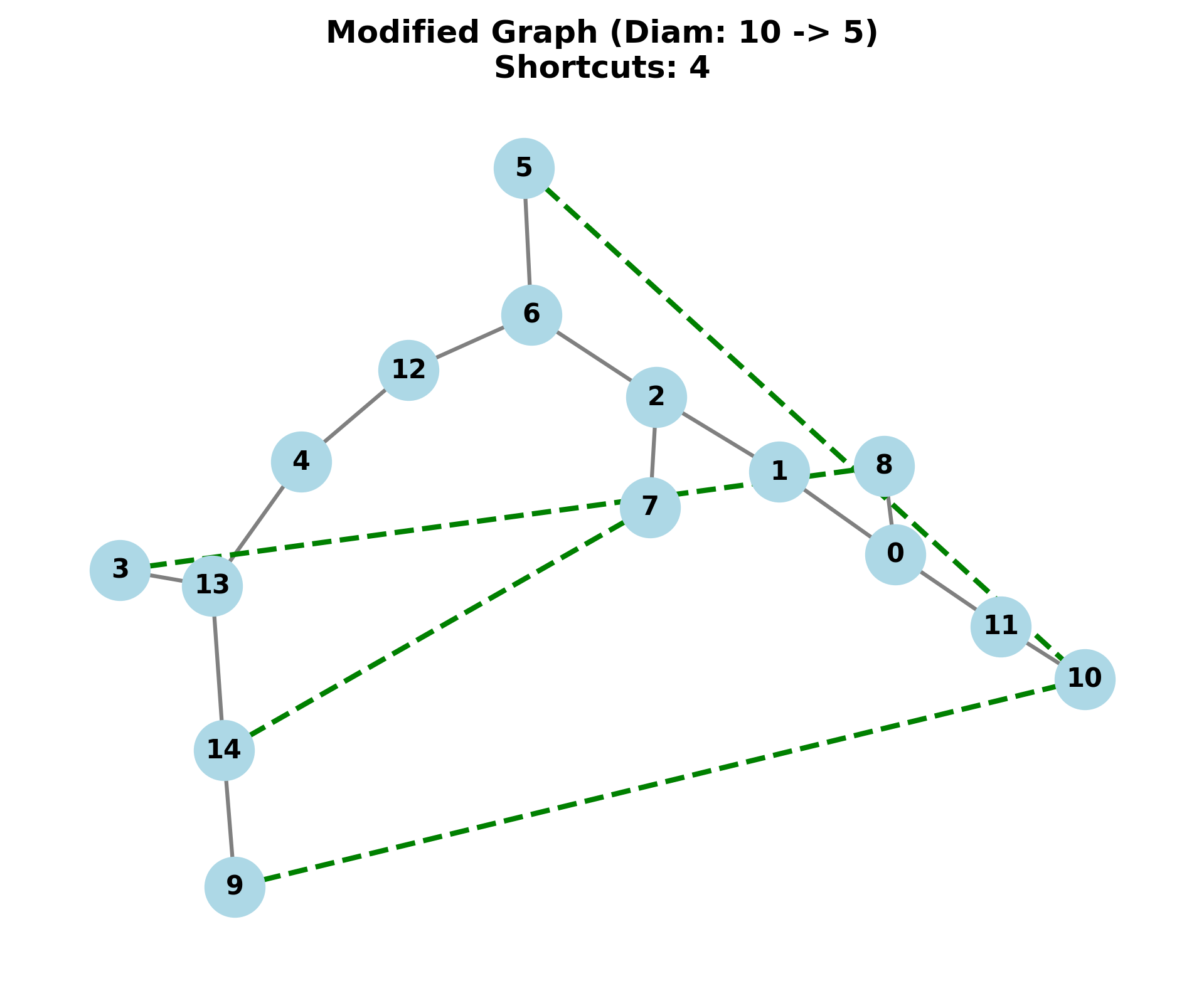}
}

\vspace{0.4cm}

\subfloat[\label{fig:reduction-d4}] {
\includegraphics[width=0.23\textwidth]{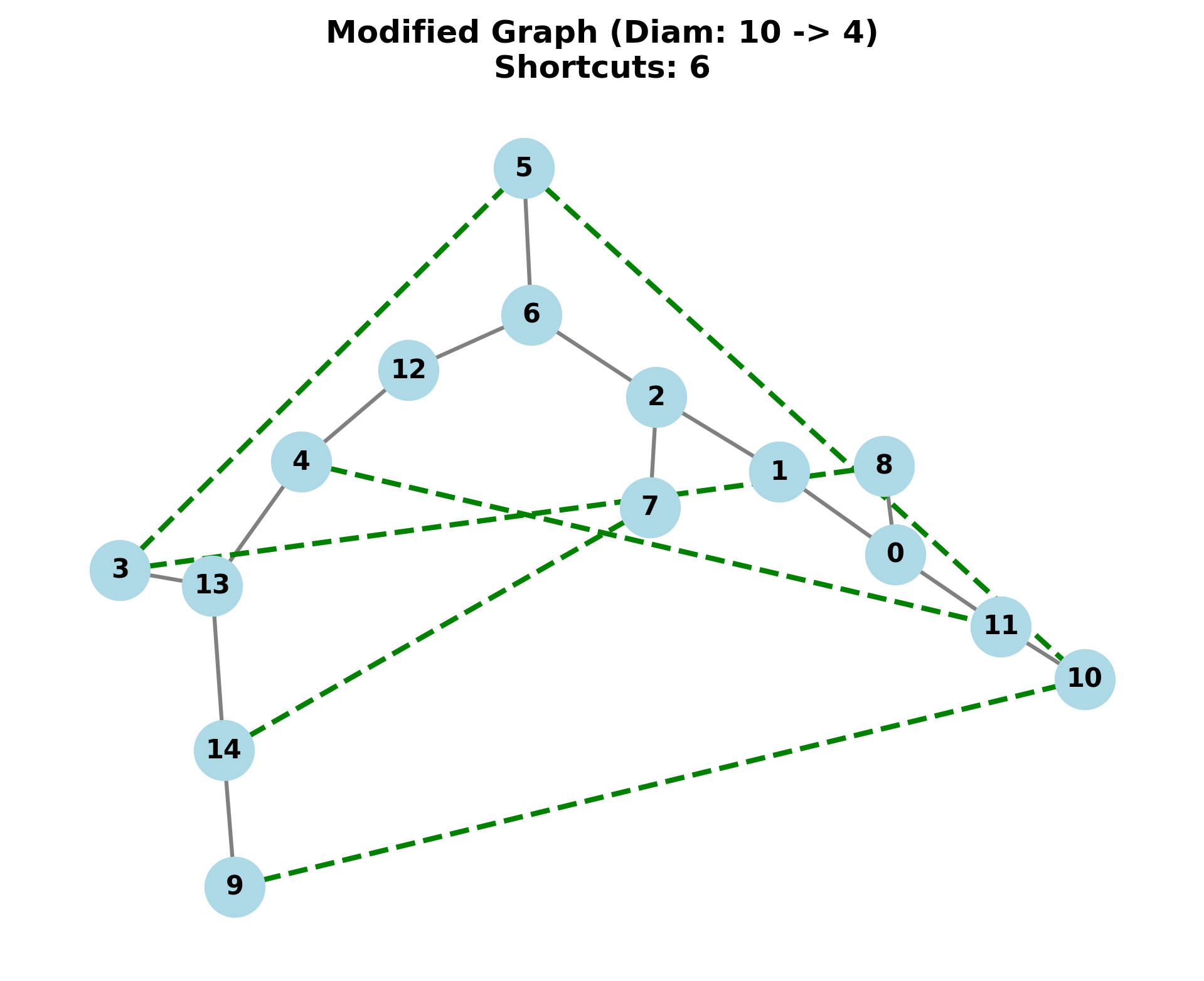}
}\hfill
\subfloat[\label{fig:reduction-d3}] {
\includegraphics[width=0.23\textwidth]{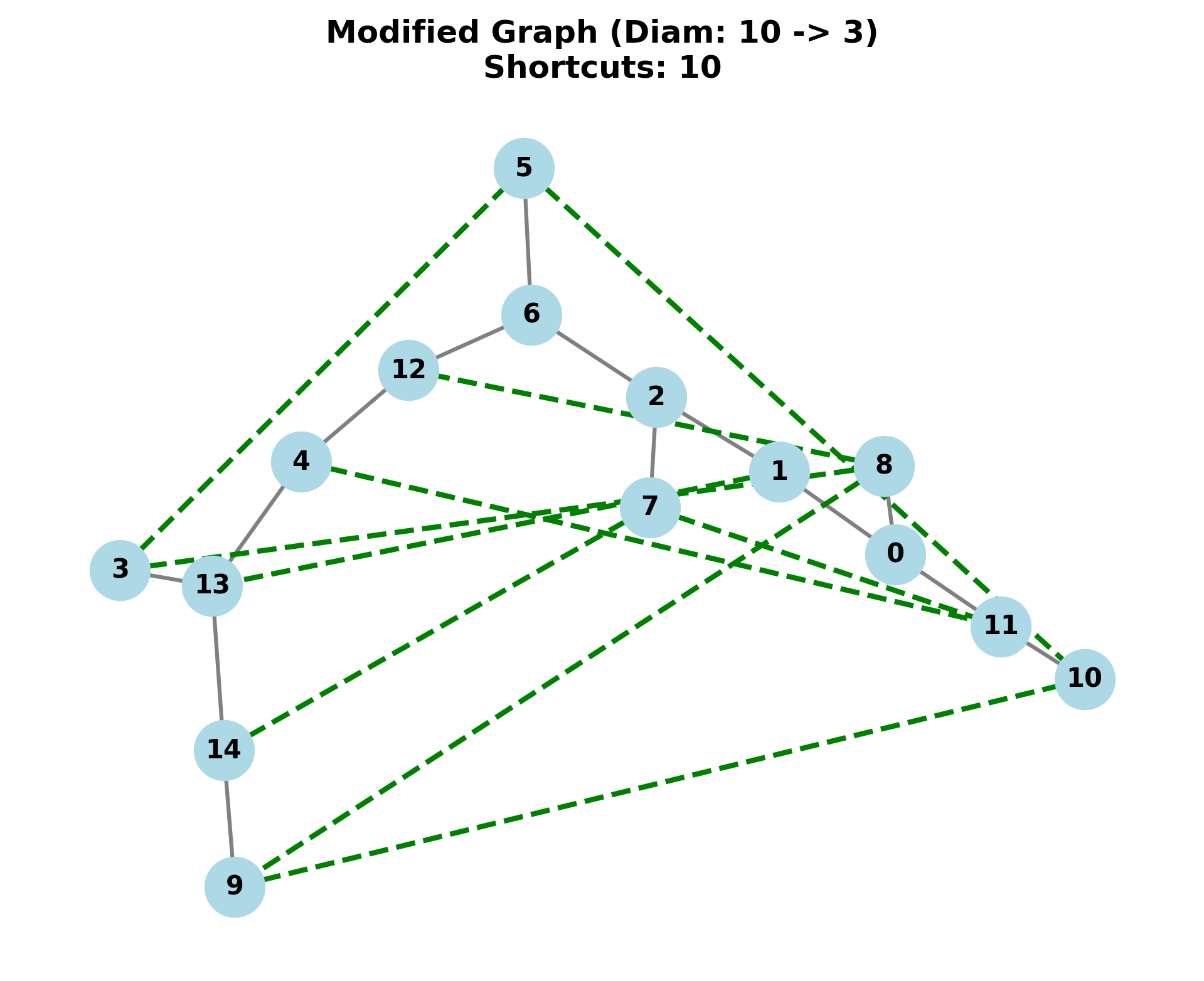}
}\hfill
\subfloat[\label{fig:reduction-d2}] {
\includegraphics[width=0.23\textwidth]{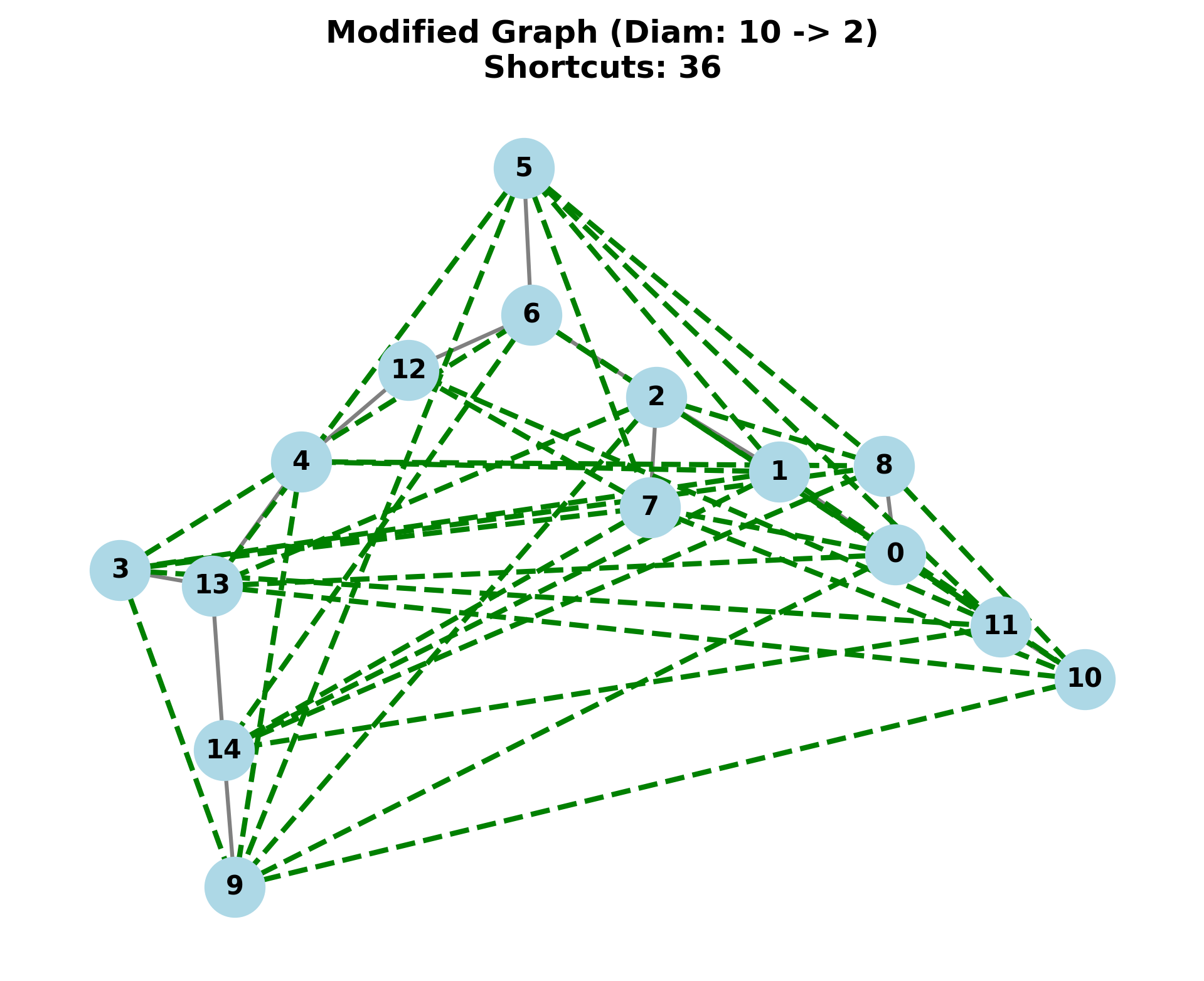}
}\hfill
\subfloat[\label{fig:reduction-d1}] {
\includegraphics[width=0.23\textwidth]{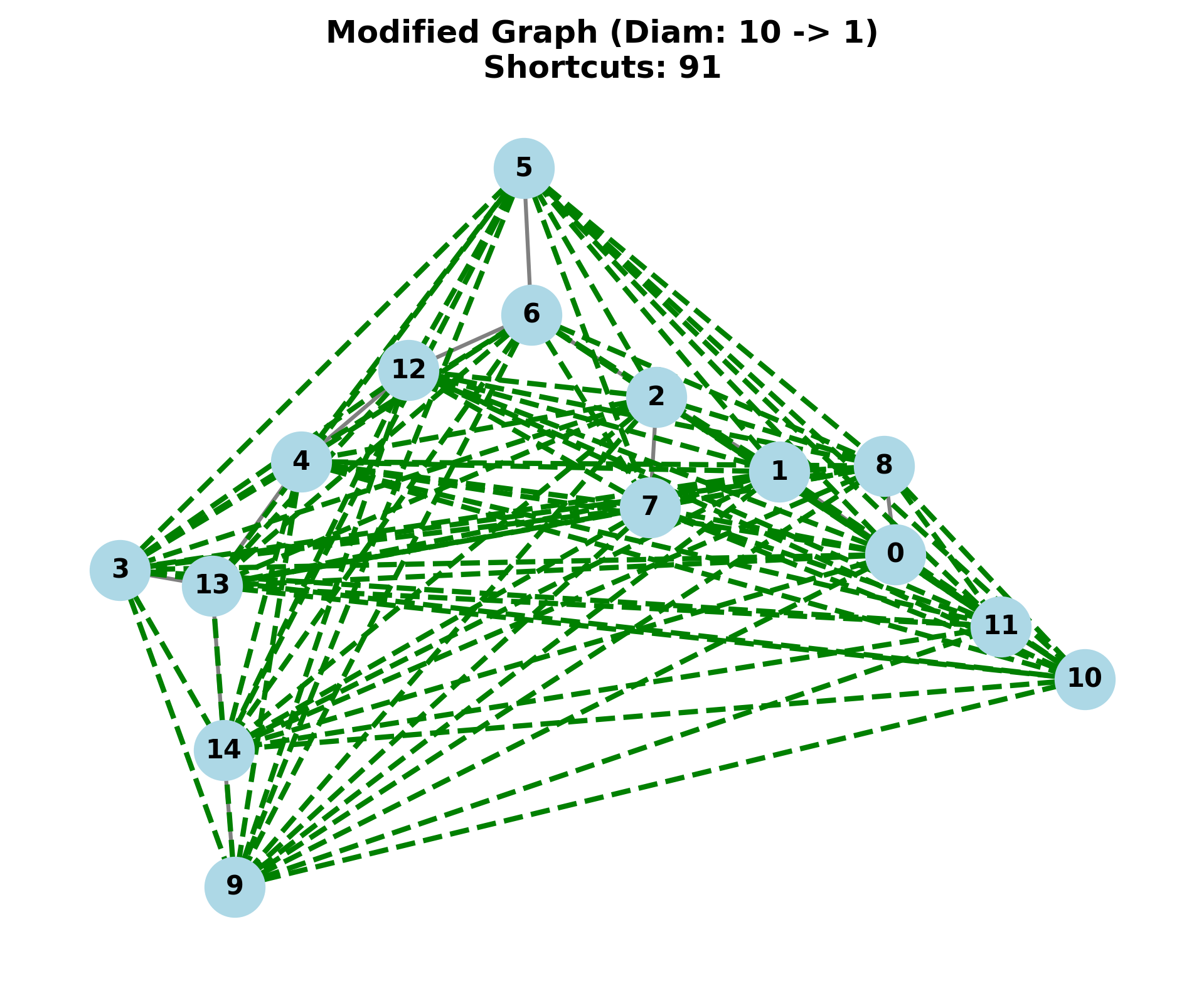}
}
\caption{Visual progression of the transport-augmentation algorithm. \textbf{(a)} The original base graph topology subject to the MaxCut cost Hamiltonian. \textbf{(b)--(h)} The modified interaction graph $G_{\mathrm{int}}$ following the iterative injection of shortcut transport couplings. This process sequentially collapses the structural diameter from $d=7$ down to $d=1$, thereby geometrically rewiring the causal lightcone to overcome depth-dependent locality obstructions without altering the underlying objective function.}
\label{fig:graph-reduction-progression}
\end{figure*}

\begin{theorem}[Optimal augmentation for depth-\(p\) locality reduction]
Fix a QAOA depth \(p\), and let \(F_p^\star\in\mathcal A_{2p}(G)\) be any optimal diameter-\(2p\) augmentation, so that \(|F_p^\star|=\alpha_{2p}(G)\). Then:
\begin{enumerate}
    \item The corresponding interaction graph \(G_p^\star:=(V,E\cup F_p^\star)\) satisfies \(\operatorname{diam}(G_p^\star)\le 2p\). Hence, by the full-lightcone criterion, for every nonempty subset \(S\subseteq V\) the graph-distance support bound for depth-\(p\) transport-augmented QAOA on \(G_p^\star\) has the entire vertex set \(V\) as its admissible lightcone.
    \item If \(F\subseteq \binom{V}{2}\setminus E\) has \(|F|<\alpha_{2p}(G)\), then the interaction graph \(G_F:=(V,E\cup F)\) satisfies \(\operatorname{diam}(G_F)>2p\). Consequently, by the necessary-depth corollary, there exist vertices \(u,v\in V\) such that every depth-\(p\) transport-augmented QAOA circuit using \(G_F\) obeys \([U_p^\dagger O_u U_p,\; O_v]=0\) for all single-site operators \(O_u\) and \(O_v\).
\end{enumerate}

In this precise sense, \(\alpha_{2p}(G)\) is the minimum number of added mixer couplings required to remove the diameter-based causal obstruction at depth \(p\).

\end{theorem}

By geometrically reducing the interaction-graph diameter, transport augmentation enlarges the Heisenberg lightcone without altering the MaxCut cost Hamiltonian. This removes the causal graph-distance obstruction at depth \(p\), allowing the shallow circuit to efficiently access global topological features that would otherwise remain causally disconnected. Although this does not immediately guarantee convergence to the global optimum, it provides a rigorous architectural foundation for the observed numerical enhancements by ensuring the transport ansatz is no longer structurally blind to long-range graph features. The formulation is graph-independent and applies to arbitrary interaction graphs; specific constructions are discussed in Appendix~\ref{app:graph-augmentation}.


\section{Experiments}

Theoretical results above are instance-independent; here we test their algorithmic implications via classical simulation of transport-augmented QAOA on graph ensembles.

\paragraph{Simulation and optimization.}
All circuits were run with Qiskit's GPU-accelerated \texttt{AerSimulator} using exact double-precision statevectors (no sampling noise). Variational angles were optimized with L-BFGS-B (function tolerance \(10^{-7}\), up to 1000 iterations). Each run used 8 random restarts with initial angles sampled uniformly from \([-0.05,0.05]\); we report the best solution across restarts. The initial state was \(|+\rangle^{\otimes N}\).

\paragraph{Ans\"atze and parameterization.}
Across all experiments, the MaxCut cost Hamiltonian remains fixed on the instance graph \(G=(V,E)\); only the mixer connectivity is modified by injecting shortcut edges \(E_{\mathrm{new}}\), yielding the interaction graph \(G_{\mathrm{int}}=(V,E\cup E_{\mathrm{new}})\). As a strong baseline we use multi-angle QAOA (ma-QAOA), assigning independent parameters to each local node/edge gate in a layer. At depth \(p=1\) this gives \(N+|E|\) parameters; adding shortcut edges gives \(N+|E|+|E_{\mathrm{new}}|\) parameters for both augmented mixers. For the scheduled \(XX+YY\) mixer, the \(XX\) and \(YY\) weights on each shortcut edge share a single parameter (parity symmetry).

\paragraph{Metrics and diameter-controlled ensembles.}
For each instance we computed the exact MaxCut optimum classically and report the approximation ratio (AR). We study uniform random 4-regular graphs, connected bipartite graphs, and random trees. For each instance we construct shortcut sets that progressively reduce \(\operatorname{diam}(G_{\mathrm{int}})\), and aggregate results by (graph family, achieved interaction diameter \(d\)) to separate diameter effects from system-size effects.

\paragraph{Main empirical trend.}
Across families, ma-QAOA remains constrained by the native topology, whereas the transport ansatz becomes nearly size-invariant once the effective interaction diameter collapses. For bipartite instances with base diameter 4, reducing the interaction path to \(d=1\) using the scheduled \(XX+YY\) mixer raises the ensemble-averaged AR from 0.7378 (ma-QAOA) to 0.9767 at \(p=1\) (\(\sigma=0.0251\), nine system sizes). For high-diameter trees, performance improves already at intermediate diameters (e.g., \(d\in\{3,4\}\)), and collapsing to \(d=1\) produces near-perfect behavior at shallow depth (see below).

The code, simulation scripts, and processed data used in our experiments are publicly available at \url{https://github.com/heyitsshoz/dealing_with_locality_in_qaoa}.

\subsection{Uniform random 4-regular instances}
Uniform random 4-regular graphs already have small base diameters (typically $d=3$ or $d=4$), so ma-QAOA at $p=1$ starts from a relatively strong baseline (AR $\approx 0.86$). Even in this regime, collapsing the interaction diameter to $d=1$ with the scheduled $XX+YY$ mixer yields near-perfect performance (ensemble mean AR $\approx 0.997$ at $p=1$). The corresponding $d=1$ scaling curves are essentially flat in $N$ (standard deviation $<0.01$ across sizes), consistent with diameter---not vertex count---setting the shallow-depth limitation.
\begin{figure}[!tbp]
\centering
\includegraphics[width=\columnwidth]{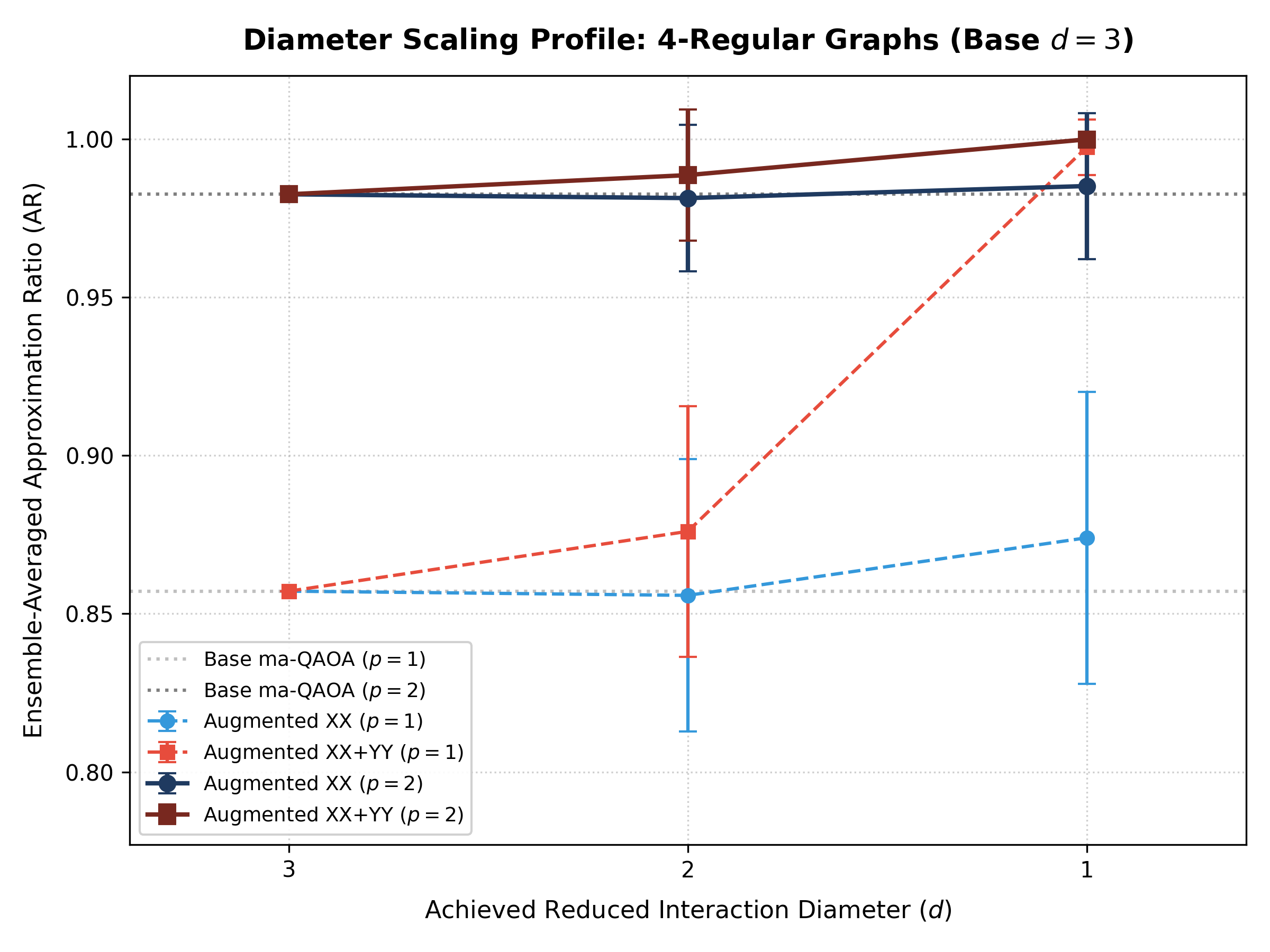}
\caption{Approximation Ratio for 4-regular instances (base diameter 3).}
\label{fig:regular-d3}
\end{figure}

\begin{figure}[!tbp]
\centering
\includegraphics[width=\columnwidth]{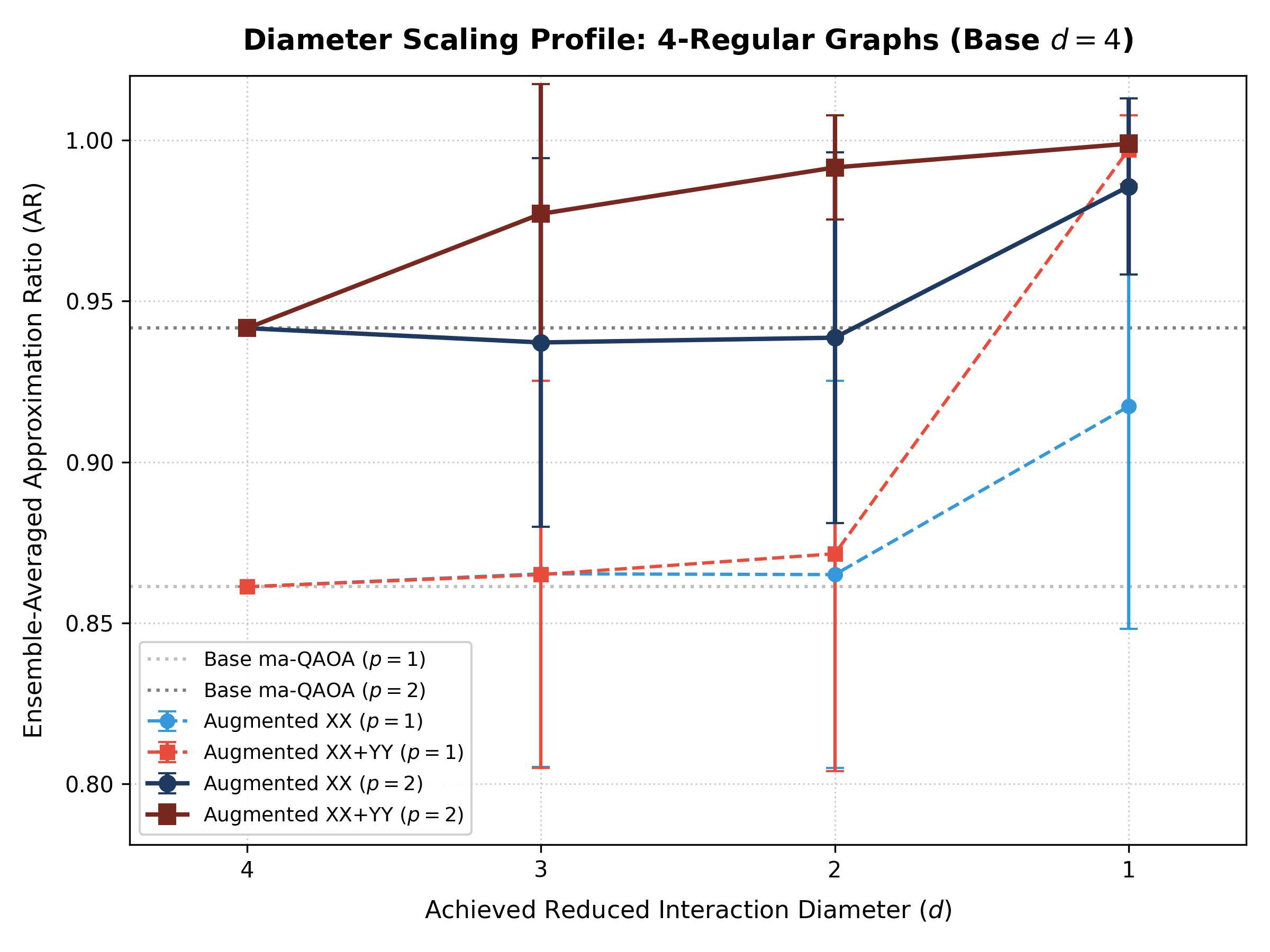}
\caption{Approximation Ratio for 4-regular instances (base diameter 4).}
\label{fig:regular-d4}
\end{figure}

\subsection{Random connected bipartite instances}
Connected bipartite instances span larger base diameters ($d=4$ to $d=9$), and the shallow-depth locality obstruction is visible directly in the $p=1$ baseline: performance degrades as diameter increases (e.g., AR $\approx 0.73$ at base diameter $d=4$). Injecting shortcuts and running the scheduled $XX+YY$ transport mixer to achieve $d=1$ produces large gains (up to $+32\%$ at $p=1$) and yields $d=1$ curves that are nearly flat in $N$, indicating that once the interaction diameter is collapsed the residual dependence on system size is weak.
\begin{figure}[!tbp]
\centering
\includegraphics[width=\columnwidth]{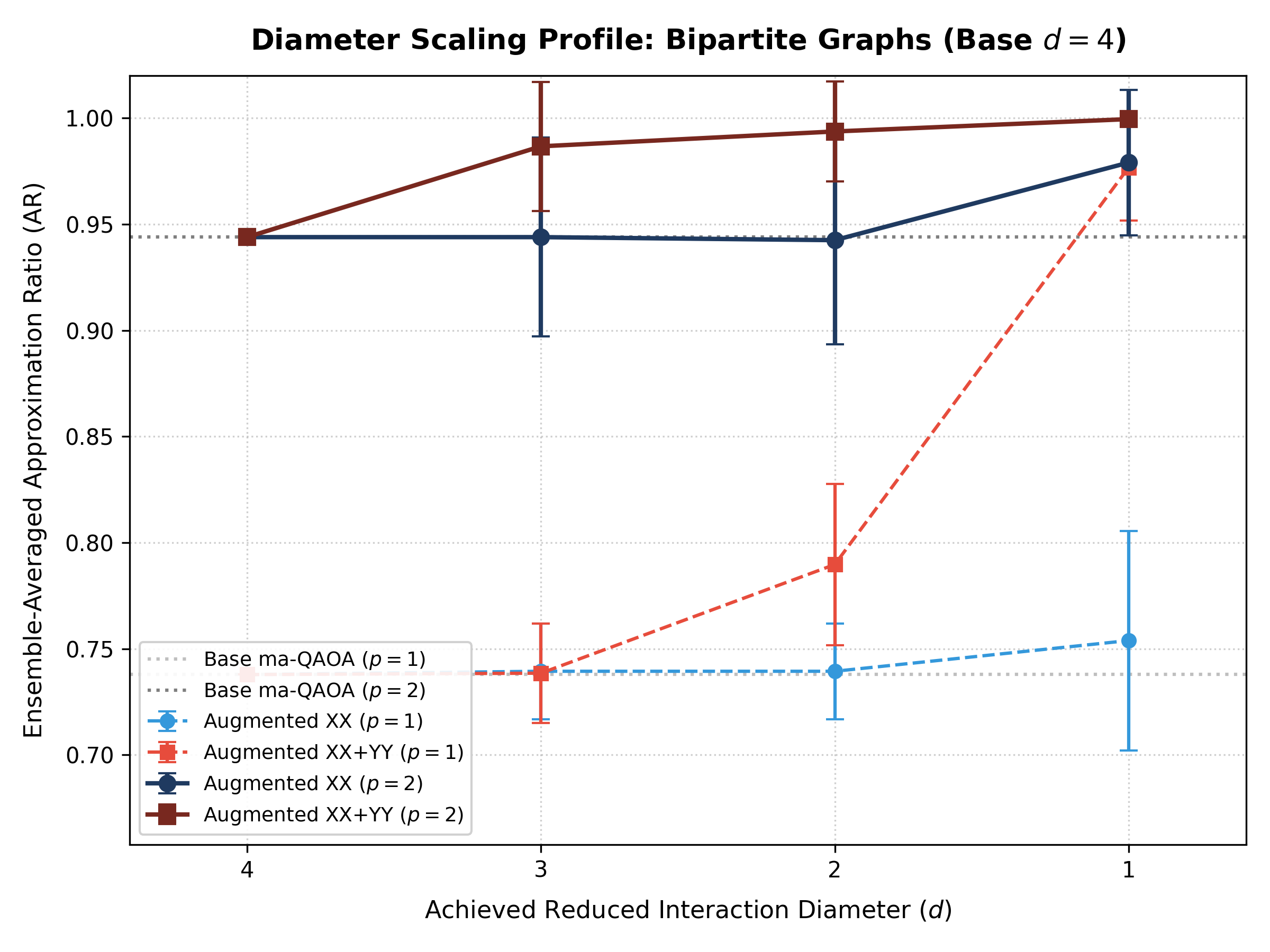}
\caption{Approximation Ratio for bipartite instances (base diameter 4).}
\label{fig:bipartite-d4}
\end{figure}

\begin{figure}[!tbp]
\centering
\includegraphics[width=\columnwidth]{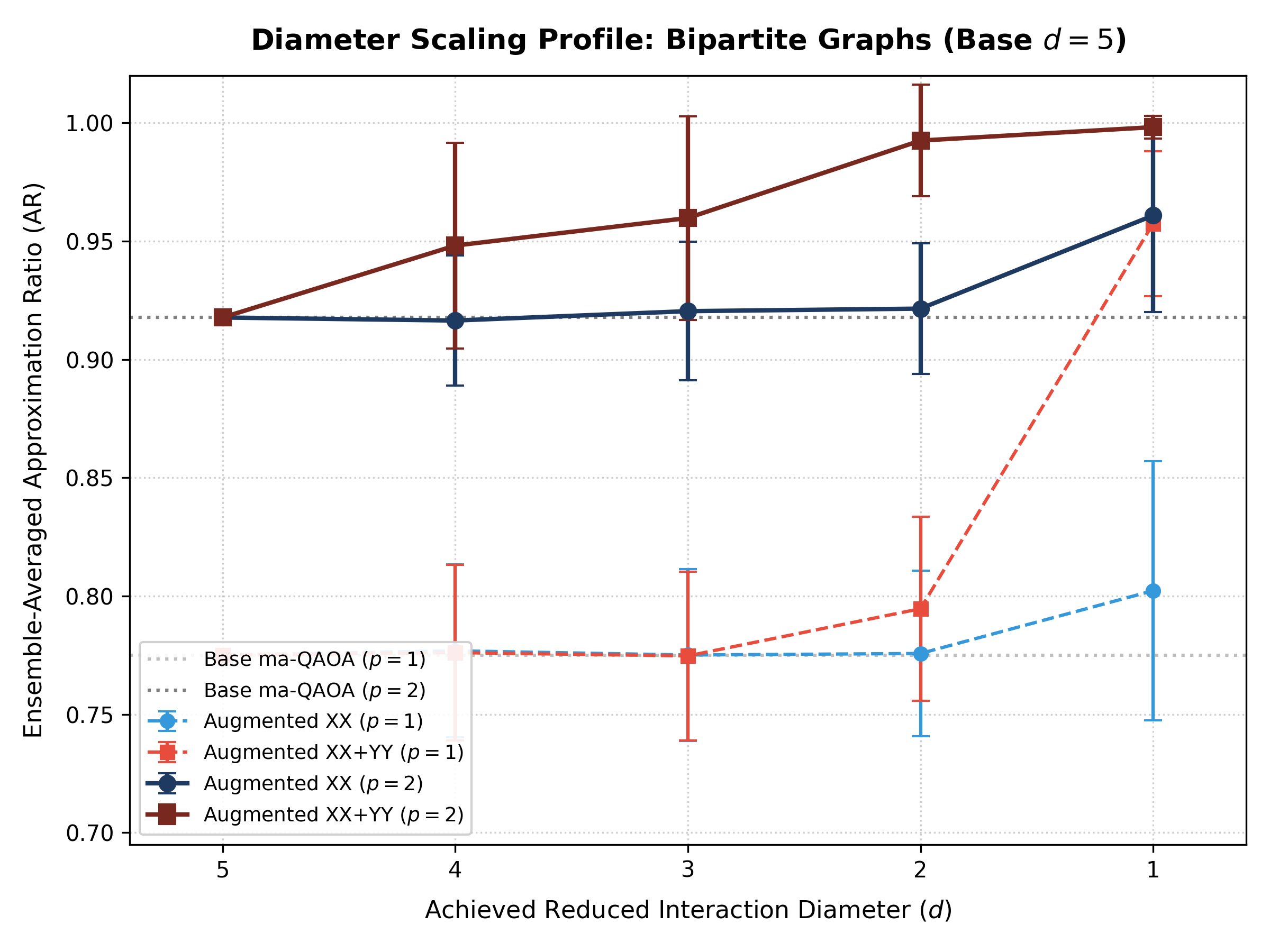}
\caption{Approximation Ratio for bipartite instances (base diameter 5).}
\label{fig:bipartite-d5}
\end{figure}

\begin{figure}[!tbp]
\centering
\includegraphics[width=\columnwidth]{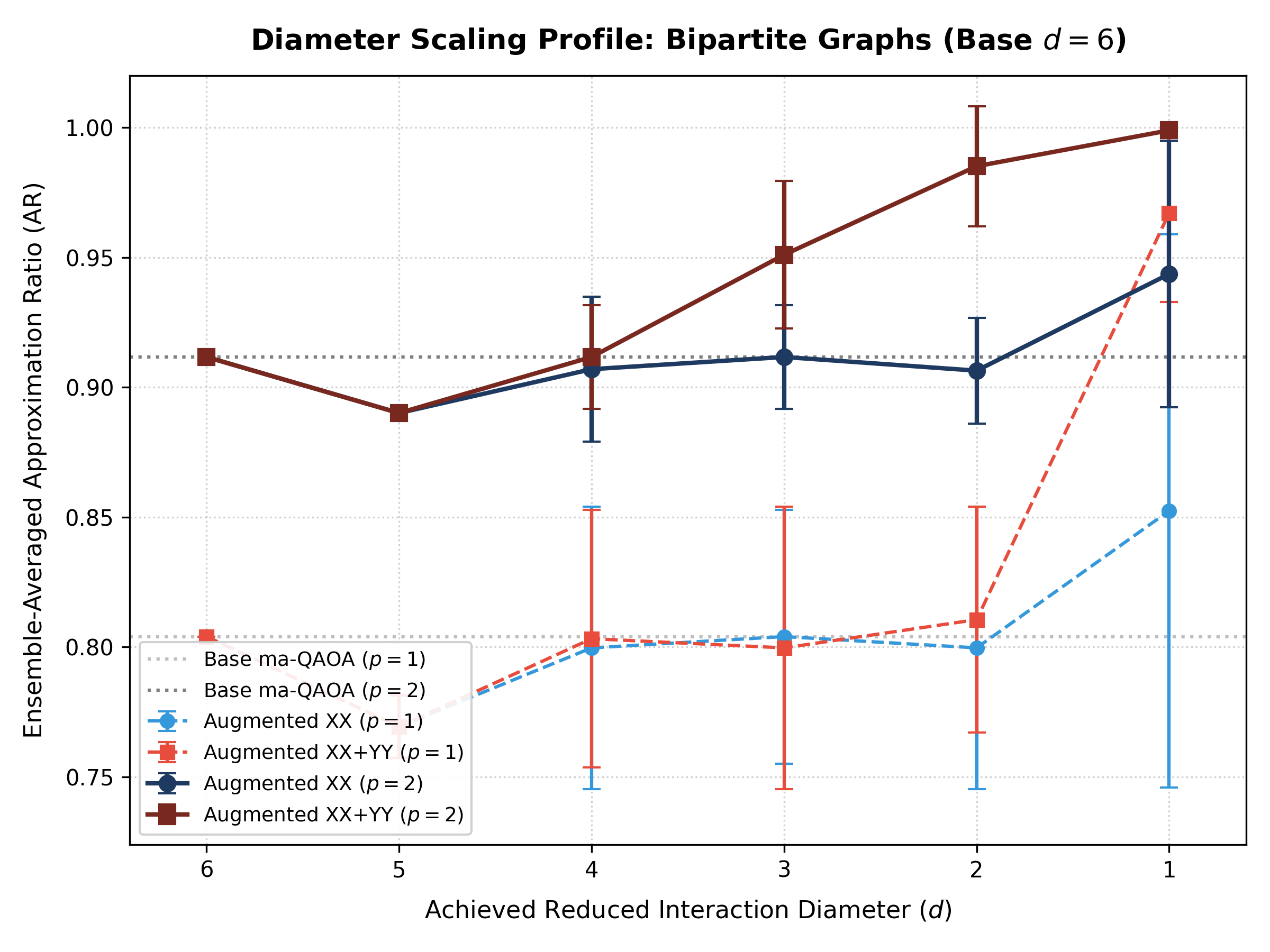}
\caption{Approximation Ratio for bipartite instances (base diameter 6).}
\label{fig:bipartite-d6}
\end{figure}

\begin{figure}[!tbp]
\centering
\includegraphics[width=\columnwidth]{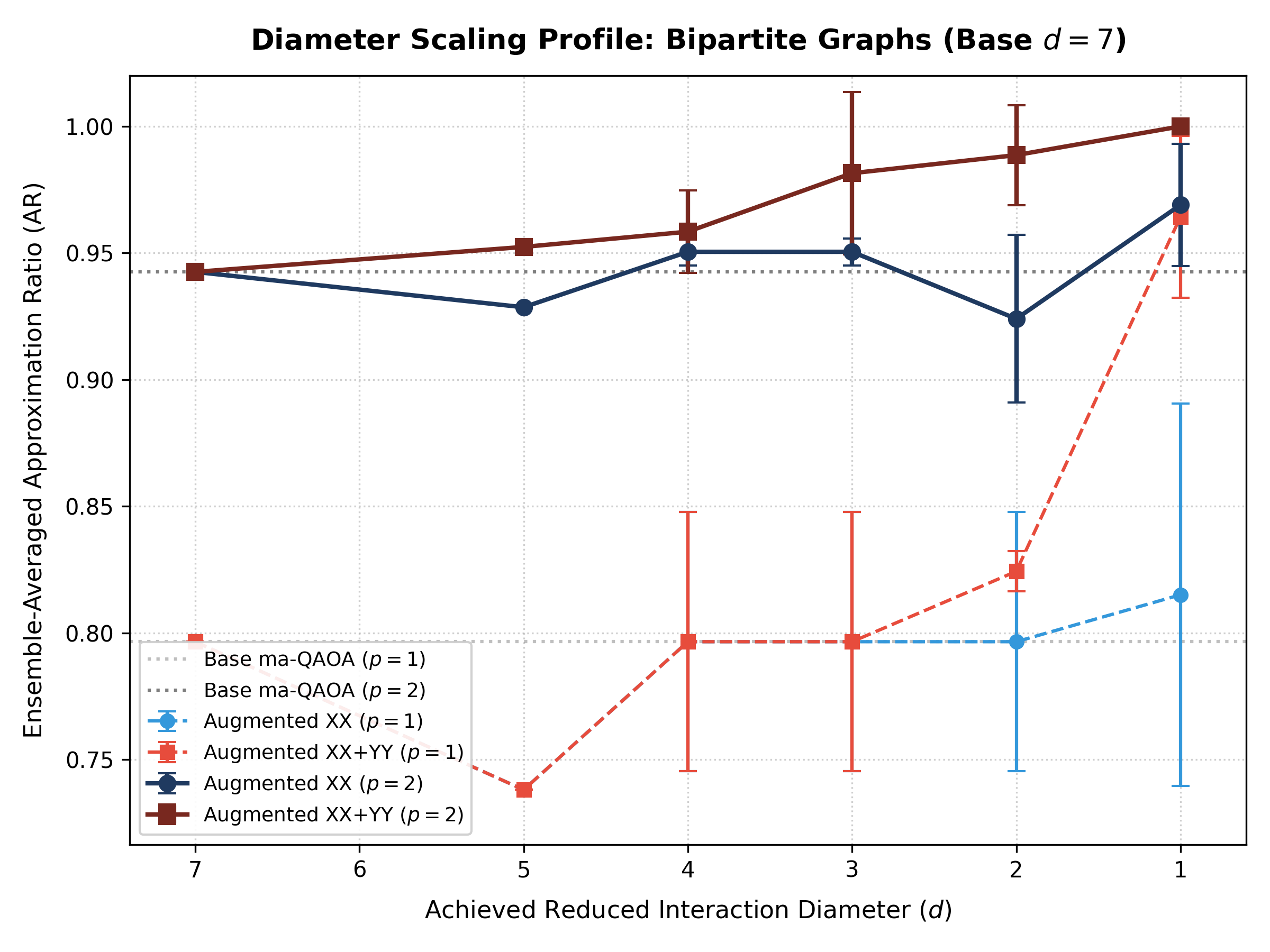}
\caption{Approximation Ratio for bipartite instances (base diameter 7).}
\label{fig:bipartite-d7}
\end{figure}

\begin{figure}[!tbp]
\centering
\includegraphics[width=\columnwidth]{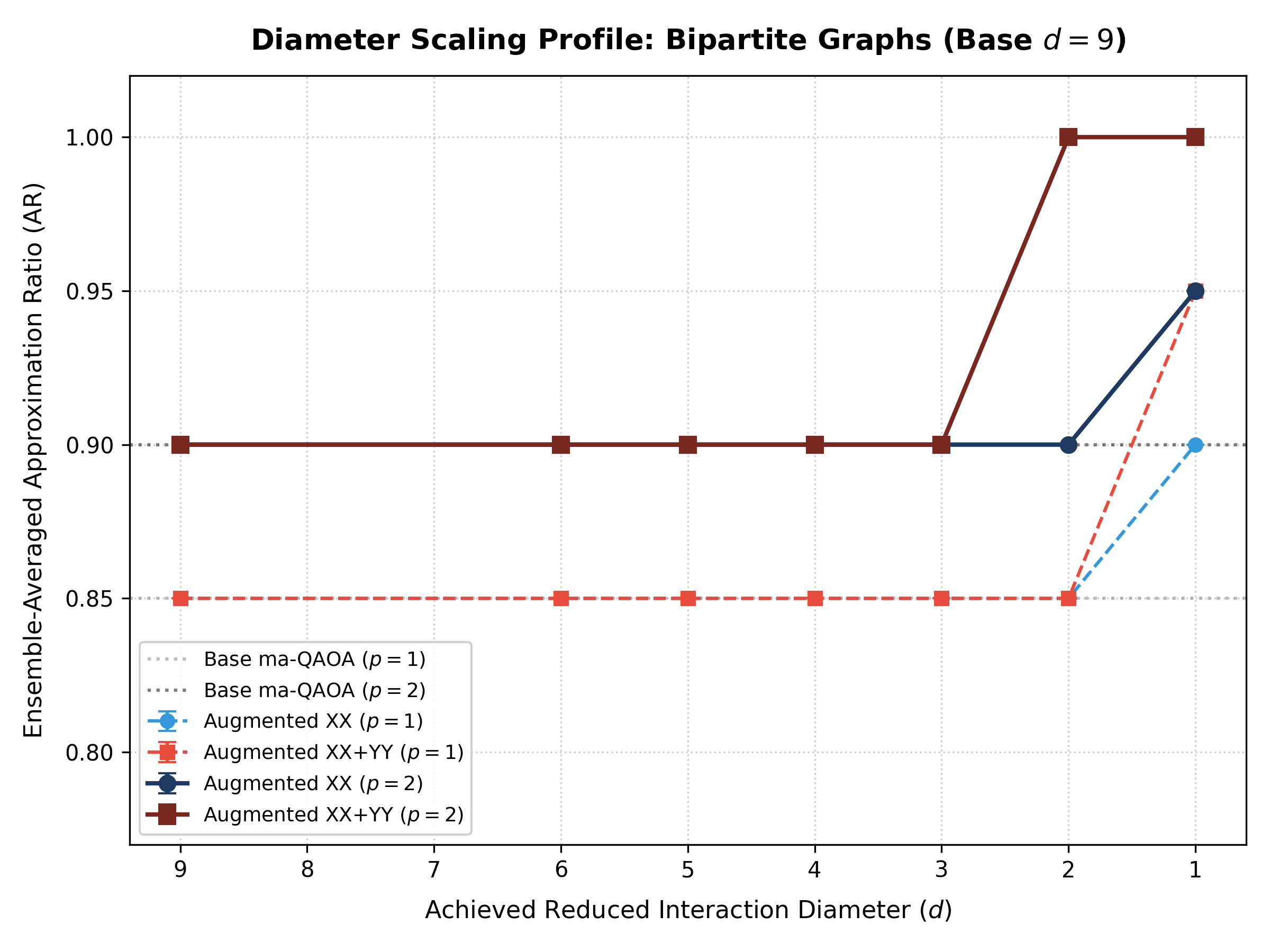}
\caption{Approximation Ratio for bipartite instances (base diameter 9).}
\label{fig:bipartite-d9}
\end{figure}

\subsection{Random tree instances}
Random trees exhibit the strongest locality obstruction, with base diameters up to $d=11$. Here the diameter dependence is most pronounced: at shallow depth the baseline struggles, but AR improves already at $p=1$ and $p=2$ once shortcuts reduce the interaction diameter to intermediate values (e.g., $d=3$ or $d=4$). Fully collapsing to $d=1$ with the scheduled $XX+YY$ mixer removes the remaining obstruction: at $p=2$ we consistently obtain perfect or near-perfect solutions (AR $\approx 1.000$) for $d=1$ across system sizes, producing nearly horizontal $AR=1$ curves in $N$.
\begin{figure}[!tbp]
\centering
\includegraphics[width=\columnwidth]{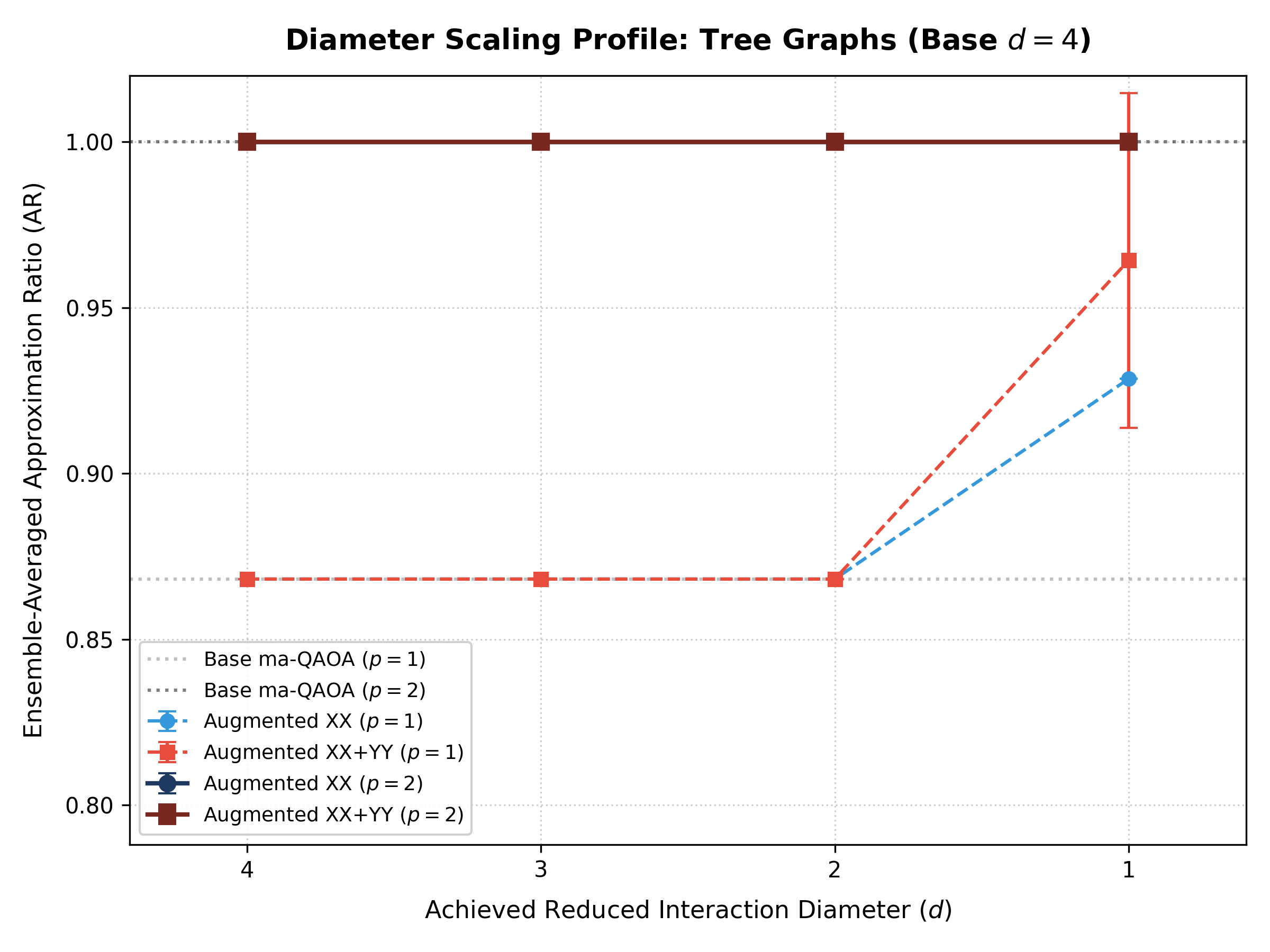}
\caption{Approximation Ratio for tree instances (base diameter 4).}
\label{fig:tree-d4}
\end{figure}
\begin{figure}[!tbp]
\centering
\includegraphics[width=\columnwidth]{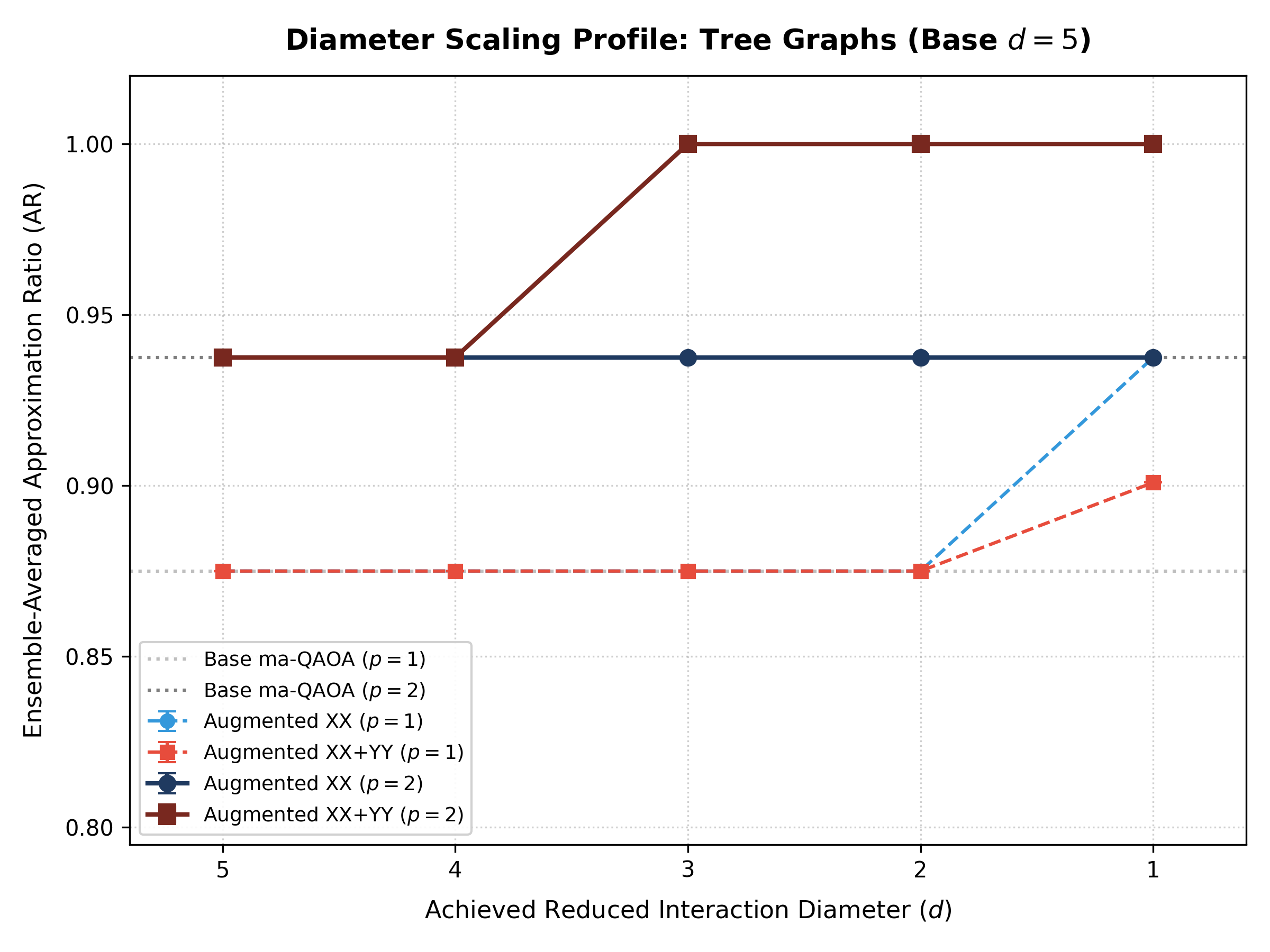}
\caption{Approximation Ratio for tree instances (base diameter 5).}
\label{fig:tree-d5}
\end{figure}
\begin{figure}[!tbp]
\centering
\includegraphics[width=\columnwidth]{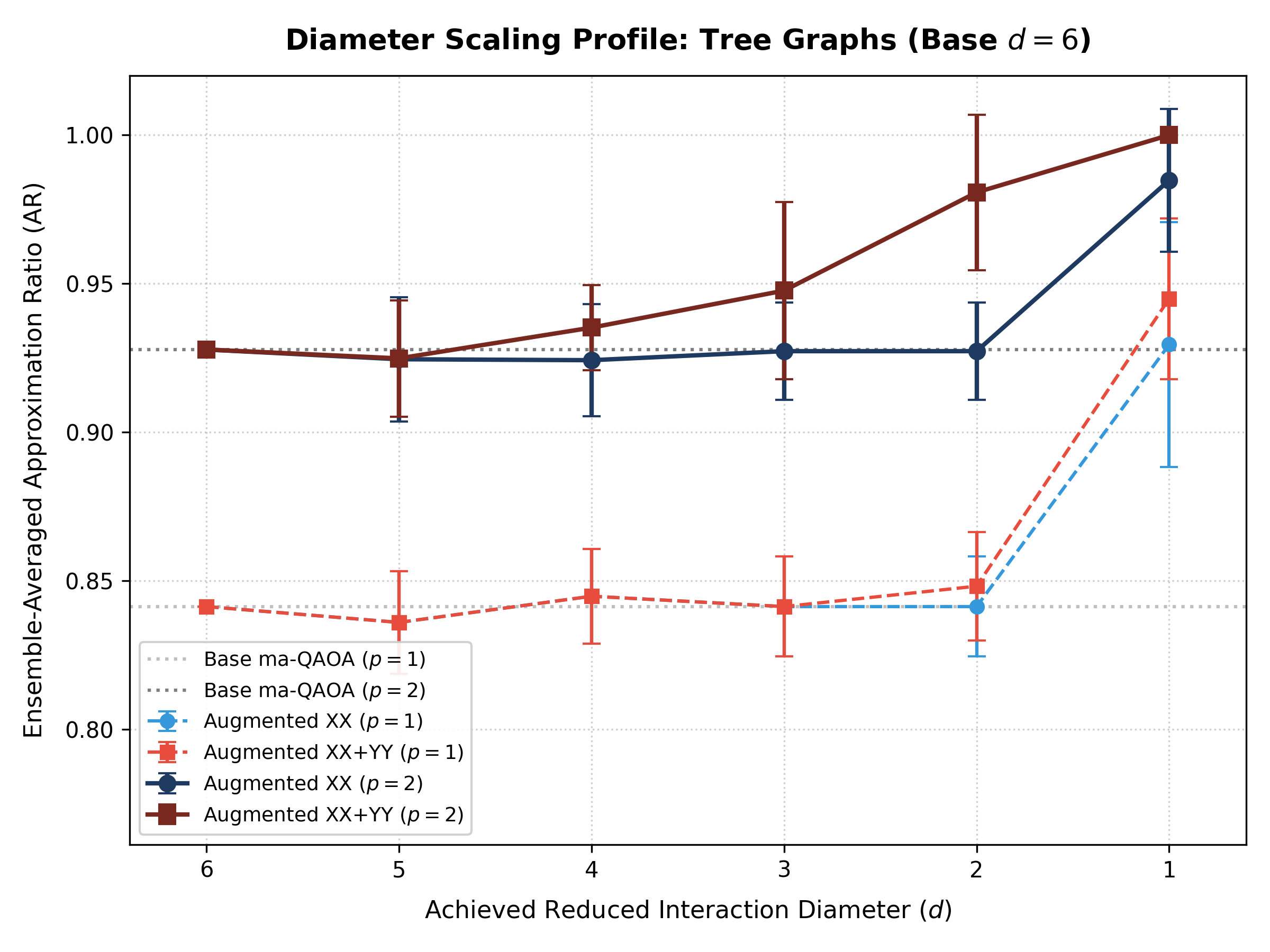}
\caption{Approximation Ratio for tree instances (base diameter 6).}
\label{fig:tree-d6}
\end{figure}
\begin{figure}[!tbp]
\centering
\includegraphics[width=\columnwidth]{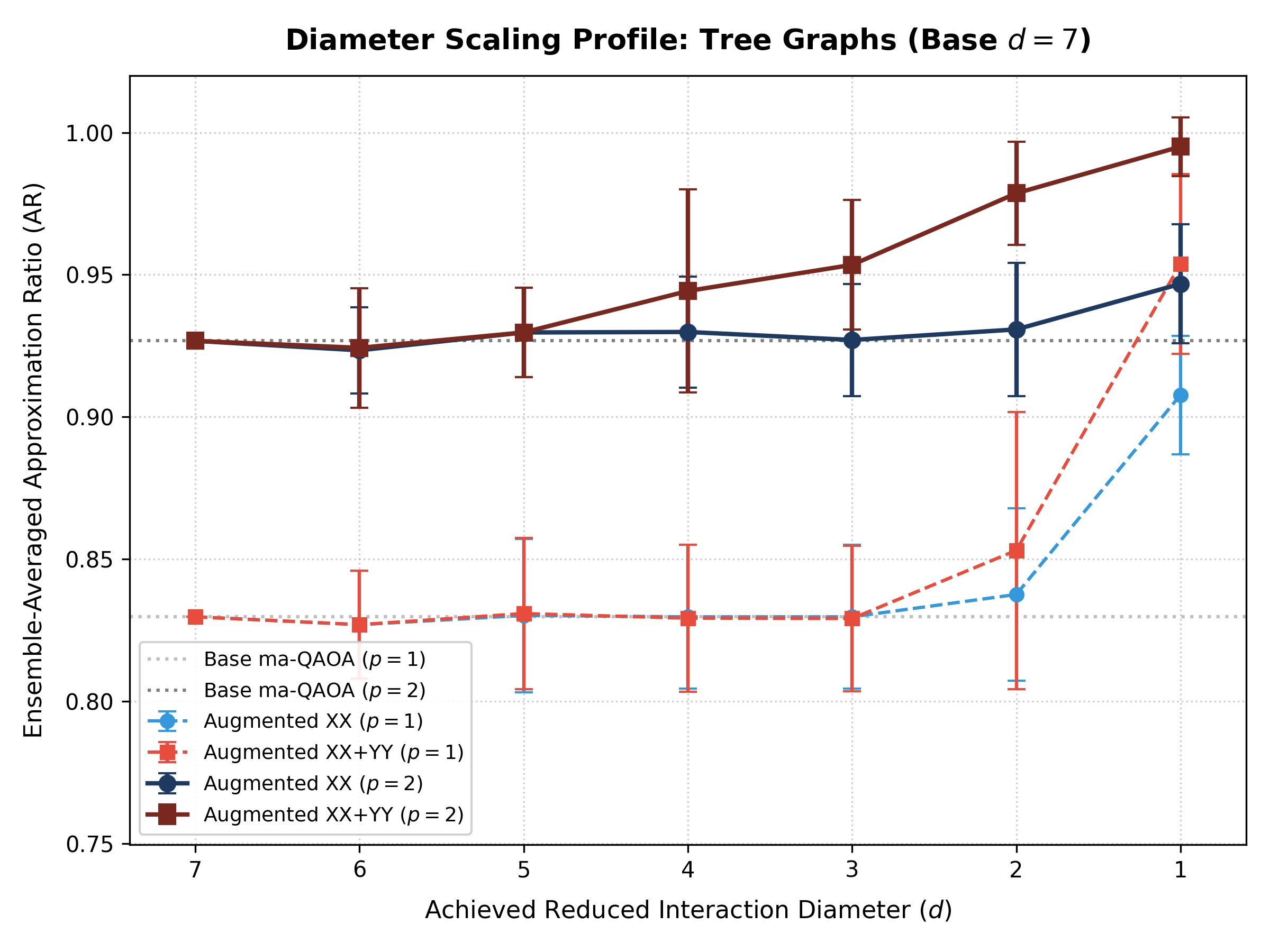}
\caption{Approximation Ratio for tree instances (base diameter 7).}
\label{fig:tree-d7}
\end{figure}
\begin{figure}[!tbp]
\centering
\includegraphics[width=\columnwidth]{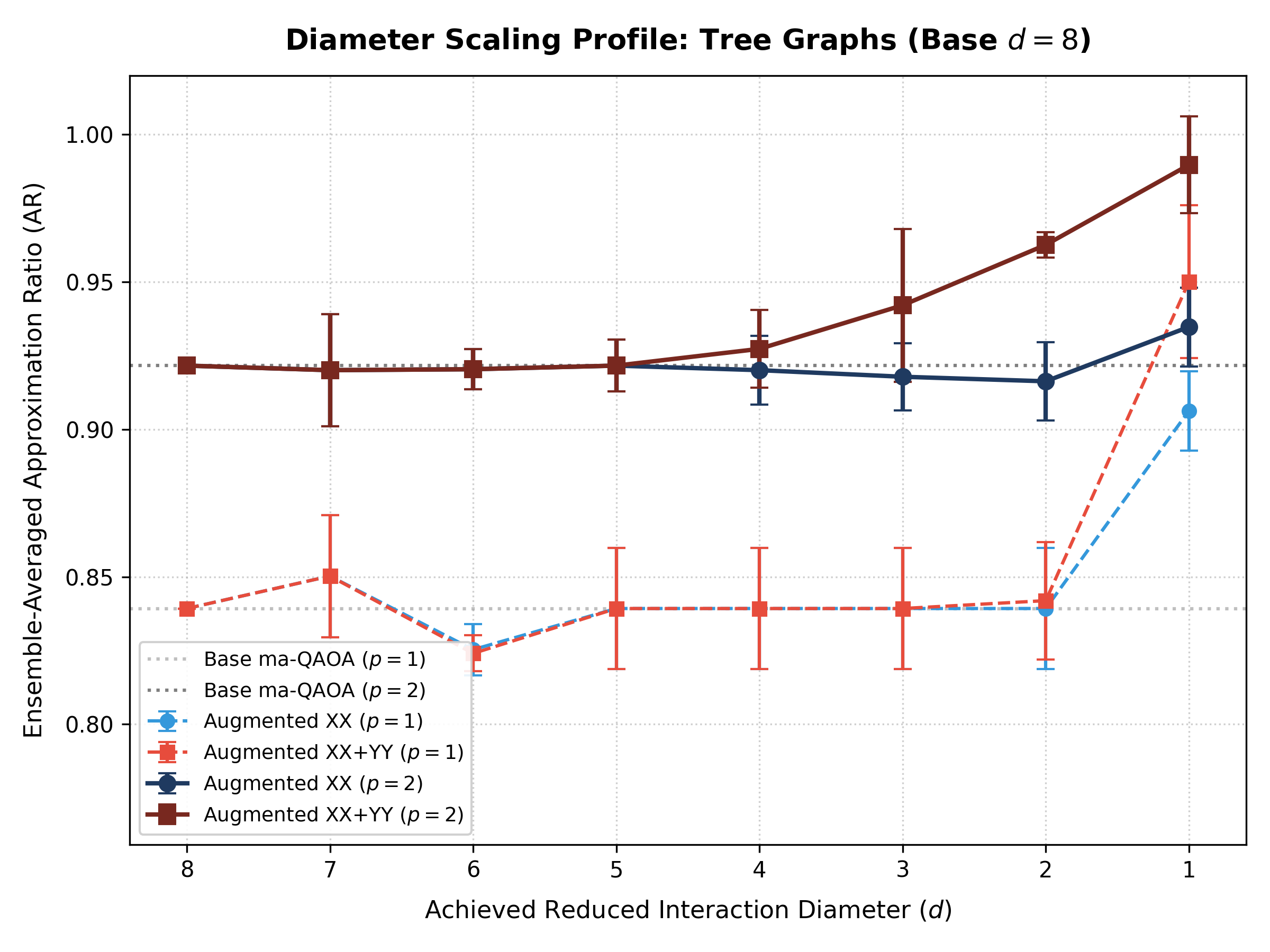}
\caption{Approximation Ratio for tree instances (base diameter 8).}
\label{fig:tree-d8}
\end{figure}
\begin{figure}[!tbp]
\centering
\includegraphics[width=\columnwidth]{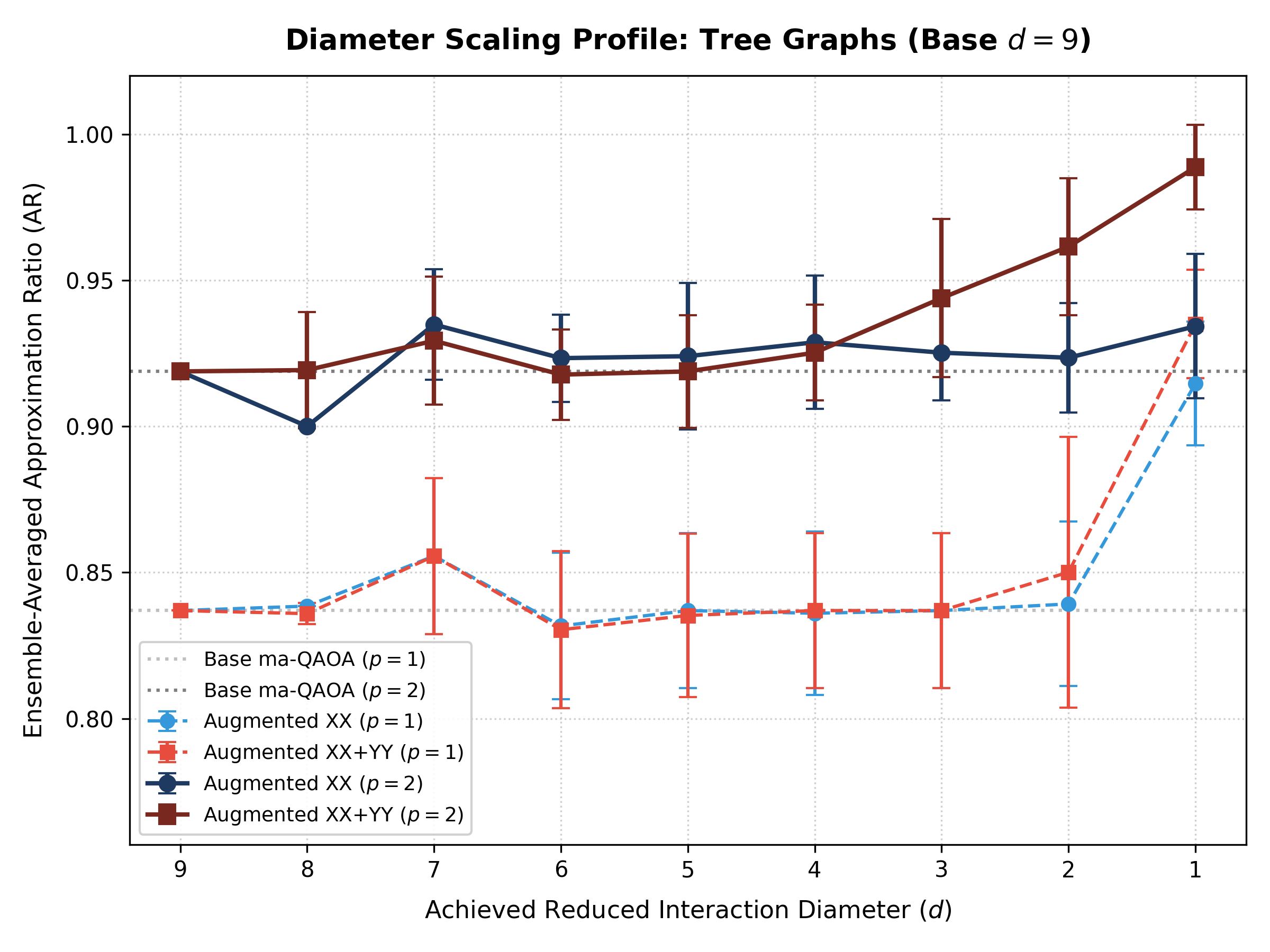}
\caption{Approximation Ratio for tree instances (base diameter 9).}
\label{fig:tree-d9}
\end{figure}
\begin{figure}[!tbp]
\centering
\includegraphics[width=\columnwidth]{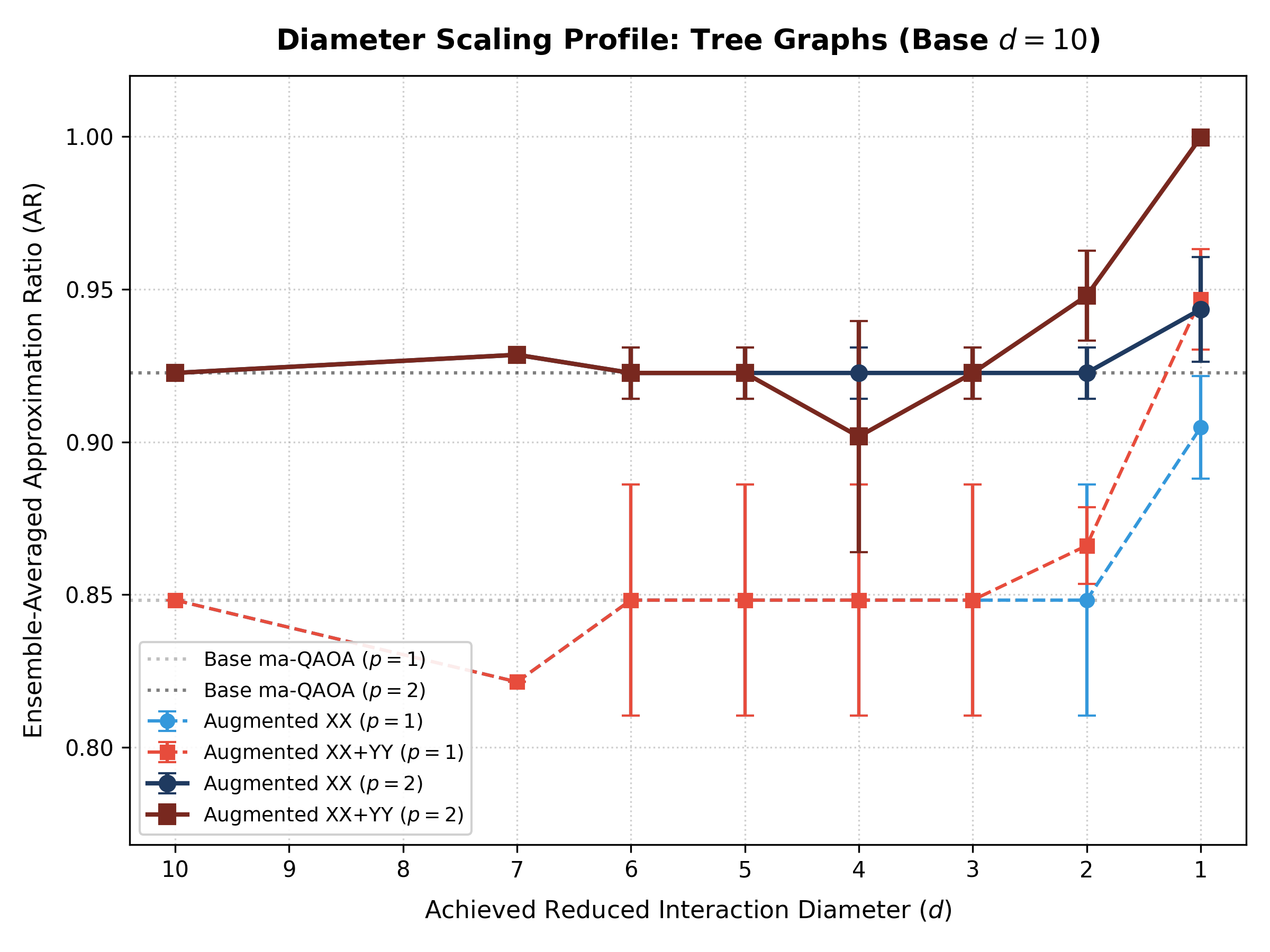}
\caption{Approximation Ratio for tree instances (base diameter 10).}
\label{fig:tree-d10}
\end{figure}
\begin{figure}[!tbp]
\centering
\includegraphics[width=\columnwidth]{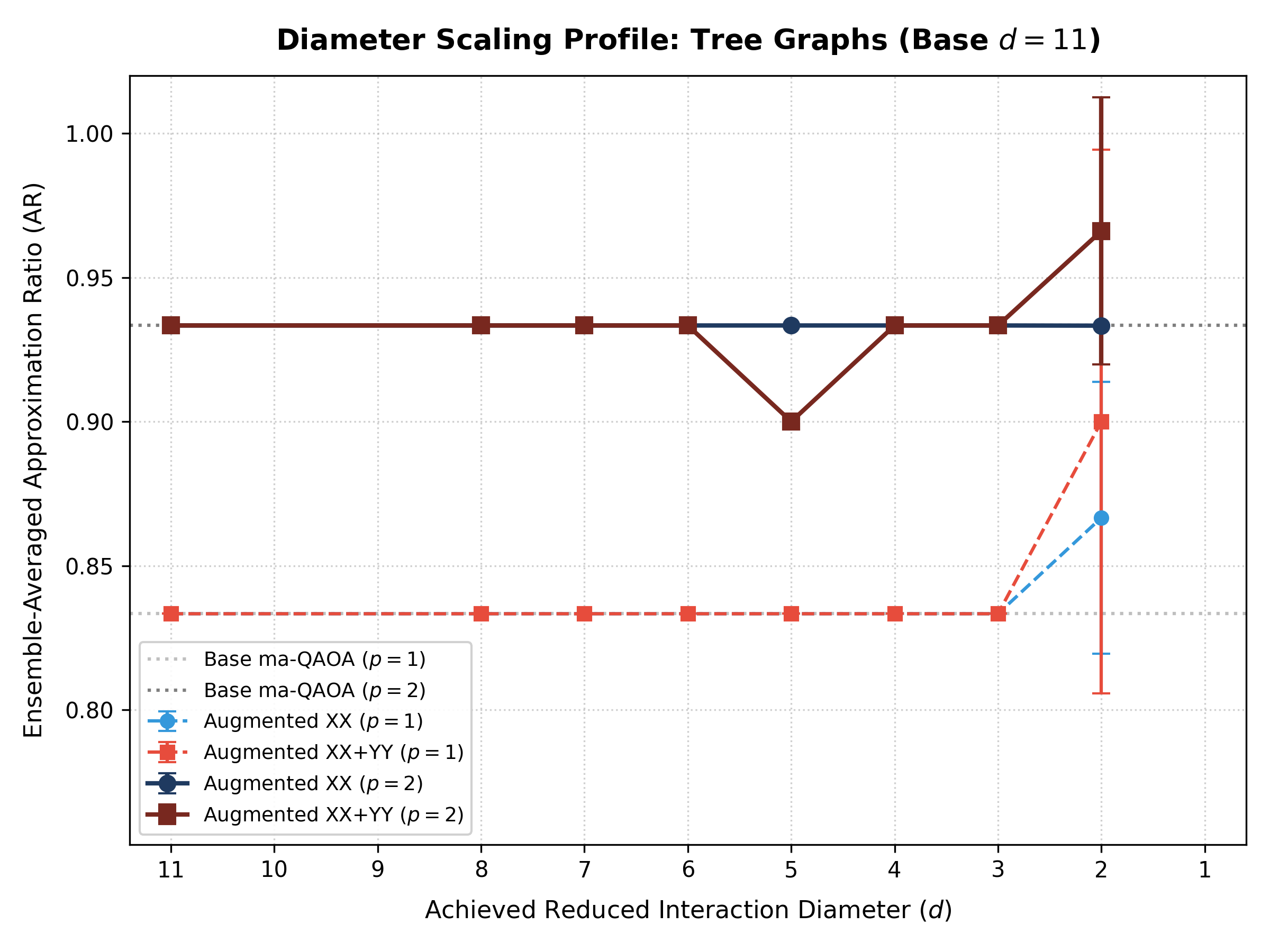}
\caption{Approximation Ratio for tree instances (base diameter 11).}
\label{fig:tree-d11}
\end{figure}

\section{Future Directions}

The transport-augmented ans\"atze studied here represent only a small subset of the possible mixer constructions that can be used to modify finite-depth QAOA locality.

A natural direction for future work is the investigation of richer transport Hamiltonians beyond the commuting \(XX\) and scheduled \(XX+YY\) mixers considered in this paper. Examples include Heisenberg-type transport interactions,
\[
\sum_{(i,j)\in E_{\mathrm{new}}}
\bigl(
J_x X_iX_j
+
J_y Y_iY_j
+
J_z Z_iZ_j
\bigr)
\]
as well as higher-order transport motifs such as hub-mediated architectures and multi-qubit couplings, for example \(XXXX\)-type interactions acting on carefully chosen vertex subsets. Additionally, relaxing the parameter constraints on the current models such as assigning independent variational angles to the \(XX\) and \(YY\) generators rather than enforcing a shared parameter could drastically increase the expressivity of the ansatz. Such constructions may further modify the causal structure of finite-depth circuits and potentially provide additional mechanisms for long-range information transport.

Another promising direction is the combination of transport augmentation with existing QAOA enhancements, including warm-start QAOA, recursive QAOA (RQAOA), adaptive-QAOA variants, problem-informed mixer constructions, and symmetry-reduced ans\"atze. Since transport augmentation primarily modifies the mixer interaction geometry while leaving the cost Hamiltonian unchanged, it is largely complementary to many of these approaches and may therefore be combined with them.

Since the number of edges and parameters are comparatively larger, one could follow a hardware-aware edge augmentation (for example nearest neighbors on chip only). A future study could involve noise models on specific hardware.

The main reason these possibilities were not explored in the present work is practical rather than conceptual. The transport-augmentation viewpoint immediately suggests a large family of potential extensions, each introducing a substantially larger design and optimization space. As observed in the higher-depth (\(p=3\)) results of this study, increasing the parameter density without careful bounding can overwhelm classical optimizers and trigger convergence failures. Evaluating such highly-parameterized variants requires not only additional numerical experiments but also careful analysis of compilation overhead, gate counts, scheduling constraints, parameter-scaling behavior, and hardware implementability. A systematic study of these tradeoffs would require significantly greater computational resources and development effort than were available within the scope of the present work.

Nevertheless, the strong empirical performance observed for the diameter-reducing transport ans\"atze considered here suggests that viewing the mixer as a designable transport geometry may provide a useful organizing principle for the development of future QAOA variants.

\section{Conclusion}

We introduced a transport-augmented QAOA framework in which the MaxCut cost Hamiltonian remains fixed on the instance graph, while the mixer is augmented with shortcut transport couplings that reduce the diameter of the circuit interaction graph.

On the theoretical side, we established exact finite-depth locality statements for commuting \(XX\) transport mixers and for scheduled \(XX+YY\) transport mixers (implemented via edge-color matchings), including explicit support-growth recursions, commutator obstructions, lightcone-volume estimates, and full-lightcone criteria. We also formalized shortcut selection as a bounded-diameter augmentation problem and connected optimal diameter-\(2p\) augmentations to the minimal mixer enlargement required to remove diameter-based causal disconnection at depth \(p\).

On the empirical side, classical statevector simulations (with multi-angle parameterization) show that performance is controlled primarily by the achieved interaction diameter rather than by the system size. Across bipartite graphs with base diameter \(d=4\), collapsing the interaction diameter to \(d=1\) with the scheduled \(XX+YY\) transport mixer increases the ensemble-averaged approximation ratio from \(0.7378\) to \(0.9767\) at depth \(p=1\), with a small spread (\(\sigma=0.0251\)) across nine system sizes. Similar diameter-driven gains appear for 4-regular graphs (near-perfect performance already at \(p=1\) once \(d=1\) is reached) and for trees, where reducing the interaction diameter yields large improvements even at shallow depths and can produce near-perfect performance when the diameter is fully collapsed.

Together, these results support a unified viewpoint: finite-depth QAOA limitations on sparse graphs are fundamentally graph-distance limited, and mixer transport augmentation provides a principled way to rewire the circuit lightcone (without changing the objective) so that shallow variational dynamics can access global graph structure.

\bibliographystyle{apsrev4-2}

\bibliography{main}

\appendix

\section{Additional Consequences of the Locality Theorem}

\begin{proposition}[Set-dynamics properties of the growth map]
For any subsets \(S,T\subseteq V\), the map \(\Phi\) is monotone and inflationary: \(S\subseteq T\implies \Phi(S)\subseteq \Phi(T)\) and \(S\subseteq \Phi(S)\). Moreover, \(N_1^{\mathrm{int}}(S)\subseteq \Phi(S)\subseteq N_2^{\mathrm{int}}(S)\). Consequently, for every integer \(t\ge0\), \(N_t^{\mathrm{int}}(S)\subseteq \Phi^t(S)\subseteq N_{2t}^{\mathrm{int}}(S)\).
The increasing chain \(S\subseteq \Phi(S)\subseteq \Phi^2(S)\subseteq \cdots\) stabilizes after at most \(|V|-|S|\) strict growth steps. If \(G_{\mathrm{int}}\) is connected and \(S\neq\emptyset\), then \(\Phi^{\operatorname{diam}(G_{\mathrm{int}})}(S)=V\).

\end{proposition}

\begin{proof}
The maps \(N_M\) and \(N_C\) are both monotone, so their composition \(\Phi=N_C\circ N_M\) is monotone as well. Also, \(S\subseteq N_M(S)\subseteq N_C(N_M(S))=\Phi(S)\),
so \(\Phi\) is inflationary.

Because \(N_C(T)\supseteq T\) for every \(T\subseteq V\), we have \(N_M(S)\subseteq \Phi(S)\).
Likewise, since \(N_M(S)\supseteq S\) and \(N_C\) is monotone, \(N_C(S)\subseteq N_C(N_M(S))=\Phi(S)\). Therefore \(N_1^{\mathrm{int}}(S)=N_C(S)\cup N_M(S)\subseteq \Phi(S)\).
The upper bound \(\Phi(S)\subseteq N_2^{\mathrm{int}}(S)\)
was already proved in the preceding theorem.

We now iterate these bounds. The base case \(t=0\) is immediate. If \(N_t^{\mathrm{int}}(S)\subseteq \Phi^t(S)\), then monotonicity of \(N_1^{\mathrm{int}}\) and the one-step lower bound give \(N_{t+1}^{\mathrm{int}}(S)=N_1^{\mathrm{int}}\!\bigl(N_t^{\mathrm{int}}(S)\bigr)\subseteq N_1^{\mathrm{int}}\!\bigl(\Phi^t(S)\bigr)\subseteq \Phi^{t+1}(S)\).
The upper inclusion \(\Phi^t(S)\subseteq N_{2t}^{\mathrm{int}}(S)\)
follows by the same induction used in the theorem.

Since \(\Phi\) is inflationary, the subsets \(\Phi^t(S)\) form an increasing chain inside the finite set \(V\). Each strict inclusion adds at least one new vertex, so there can be at most \(|V|-|S|\) strict growth steps before stabilization.

Finally, if \(G_{\mathrm{int}}\) is connected and \(S\neq\emptyset\), then every vertex lies within distance at most \(\operatorname{diam}(G_{\mathrm{int}})\) of \(S\). Hence \(N_{\operatorname{diam}(G_{\mathrm{int}})}^{\mathrm{int}}(S)=V\). Using the lower bound already proved, \(V = N_{\operatorname{diam}(G_{\mathrm{int}})}^{\mathrm{int}}(S) \subseteq \Phi^{\operatorname{diam}(G_{\mathrm{int}})}(S) \subseteq V\), which forces \(\Phi^{\operatorname{diam}(G_{\mathrm{int}})}(S)=V\).

\end{proof}

\begin{corollary}[Locality reduction through mixer augmentation]

Let \(G_{\mathrm{int}}=(V,E\cup E_{\mathrm{new}})\) be the full interaction graph associated with the transport-augmented circuit.

Then the support of the Heisenberg-evolved observable \(C_{uv}^{(p)}\) is contained entirely within the radius-\(2p\) neighborhood of the set \(\{u,v\}\) in the interaction graph \(G_{\mathrm{int}}\). That is, \(\operatorname{supp}(C_{uv}^{(p)}) \subseteq \Phi^p(\{u,v\}) \subseteq N_{2p}^{\mathrm{int}}(\{u,v\})\), or equivalently, \(\operatorname{supp}(C_{uv}^{(p)}) \subseteq B_{G_{\mathrm{int}}}(u,2p)\cup B_{G_{\mathrm{int}}}(v,2p)\).

\end{corollary}

\begin{proof}

Apply the theorem with \(S=\{u,v\}\) and \(O_S=C_{uv}\).

\end{proof}

\begin{corollary}[Schr\"odinger-picture objective locality]
For the commuting \(XX\) transport mixer, fix a cost edge \((u,v)\in E\) and define the exact causal region \(\mathcal L_{uv}^{XX}:=\Phi^p(\{u,v\})\). Write \(\mathcal L=\mathcal L_{uv}^{XX}\), and let \(U_{p,\mathcal L}^{XX}\) denote the depth-\(p\) circuit obtained by keeping only those single-qubit, mixer-edge, and cost-edge gates whose supports are contained in \(\mathcal L\), with the same layer order and the same angles as in the full circuit. Then
\(\langle s|U_p^\dagger C_{uv}U_p|s\rangle=\langle s_{\mathcal L}|(U_{p,\mathcal L}^{XX})^\dagger C_{uv}U_{p,\mathcal L}^{XX}|s_{\mathcal L}\rangle\),
where \(|s_{\mathcal L}\rangle\) is the product initial state restricted to \(\mathcal L_{uv}^{XX}\). Consequently, the local objective contribution \(\langle C_{uv}\rangle\) is a function only of the restricted causal subcircuit inside \(\mathcal L_{uv}^{XX}\subseteq N_{2p}^{\mathrm{int}}(\{u,v\})\).

For the scheduled \(XX+YY\) mixer with \(q\) edge-color sublayers, the same statement holds with \(\mathcal L_{uv}^{XY}:=\Psi_q^p(\{u,v\})\subseteq N_{(q+1)p}^{\mathrm{int}}(\{u,v\})\). If some variational angles are tied globally across all edges in a layer, those shared scalar parameters appear in every local function; the locality statement is that distant gates and graph data outside the causal region do not otherwise affect the local contribution.
\end{corollary}

\begin{proof}
We prove the \(XX\) statement; the scheduled \(XX+YY\) statement follows from the same argument with \(\Phi\) replaced by \(\Psi_q\). The support-recursion proof does more than bound the final support. In the Heisenberg conjugation of \(C_{uv}\), every gate whose support is disjoint from the current support commutes through and cancels against its adjoint. The only gates that survive are those incident on the recursively generated sets \(\{u,v\},\ \Phi(\{u,v\}),\ \dots,\ \Phi^p(\{u,v\})\). Since the final region \(\mathcal L_{uv}^{XX}=\Phi^p(\{u,v\})\) is inflationary and contains all earlier recursively generated sets, every surviving gate has support contained in \(\mathcal L_{uv}^{XX}\). Therefore \(U_p^\dagger C_{uv}U_p=\bigl((U_{p,\mathcal L}^{XX})^\dagger C_{uv}U_{p,\mathcal L}^{XX}\bigr)\otimes \mathbb I_{V\setminus \mathcal L_{uv}^{XX}}\). Because the initial state factorizes as \(|s\rangle=|s_{\mathcal L}\rangle\otimes |s_{V\setminus \mathcal L}\rangle\), taking the expectation value gives the claimed restricted-circuit formula.

\end{proof}

\begin{corollary}[Lightcone-volume bound]
For any operator \(O_S\) initially supported on \(S\subseteq V\), \(|\operatorname{supp}(U_p^\dagger O_S U_p)| \le |\Phi^p(S)| \le |N_{2p}^{\mathrm{int}}(S)|\).
\end{corollary}

\begin{proof}
This follows immediately from \(\operatorname{supp}(U_p^\dagger O_S U_p)\subseteq \Phi^p(S)\subseteq N_{2p}^{\mathrm{int}}(S)\).
\end{proof}

\begin{corollary}[Bounded-degree lightcone-volume estimate]
Assume the interaction graph \(G_{\mathrm{int}}\) has maximum degree at most \(\Delta\ge2\). Then for every operator \(O_S\) initially supported on \(S\subseteq V\),
\[
|\operatorname{supp}(U_p^\dagger O_S U_p)|
\le
|S|
\left(
1+\Delta\sum_{r=0}^{2p-1}(\Delta-1)^r
\right)
\]
Equivalently,
\[
|\operatorname{supp}(U_p^\dagger O_S U_p)|
\le
|S|(1+4p)
\qquad
\text{if }\Delta=2
\]
and
\begin{multline*}
|\operatorname{supp}(U_p^\dagger O_S U_p)|
\le\\
|S|
\left(
1+\Delta\frac{(\Delta-1)^{2p}-1}{\Delta-2}
\right)
\qquad
\text{if }\Delta\ge3
\end{multline*}
\end{corollary}

\begin{proof}
By the lightcone-volume bound, it is enough to estimate \(|N_{2p}^{\mathrm{int}}(S)|\). For a single vertex \(u\in V\), the ball of radius \(2p\) contains at most
\[
1+\Delta\sum_{r=0}^{2p-1}(\Delta-1)^r
\]
vertices, because after the first step there are at most \(\Delta\) choices and at each later step there are at most \(\Delta-1\) new branches. Summing this crude bound over \(u\in S\) gives
\[
|N_{2p}^{\mathrm{int}}(S)|
\le
|S|
\left(
1+\Delta\sum_{r=0}^{2p-1}(\Delta-1)^r
\right)
\]
Combining this with \(|\operatorname{supp}(U_p^\dagger O_S U_p)|\le |N_{2p}^{\mathrm{int}}(S)|\) proves the result. The displayed closed forms are the corresponding evaluations of the geometric sum.
\end{proof}

\begin{corollary}[Necessary depth for nontrivial operator influence]
If \(\operatorname{diam}(G_{\mathrm{int}})>2p\), then there exist vertices \(u,v\in V\) such that for every pair of operators \(O_u\) and \(O_v\) supported on \(\{u\}\) and \(\{v\}\), respectively, we have \([U_p^\dagger O_u U_p,\; O_v]=0\).
\end{corollary}

\begin{proof}
Choose \(u,v\in V\) with \(\mathrm{dist}_{G_{\mathrm{int}}}(u,v)>2p\). Apply the exact commutator corollary with \(S=\{u\}\) and \(T=\{v\}\).
\end{proof}

\begin{proposition}[Sharpness of the radius-\(2p\) bound]
The factor \(2\) in the interaction-graph lightcone bound is optimal in general. For every depth \(p\), there exist a transport-augmented instance, a local observable, and parameter choices for which the Heisenberg-evolved observable has support reaching a vertex at interaction-graph distance exactly \(2p\) from the initial support.

\end{proposition}

\begin{proof}
Consider vertices
\(v_0,v_1,\dots,v_{2p+1}\), cost edges \(E=\{(v_{2r},v_{2r+1}):0\le r\le p\}\), and transport edges \(E_{\mathrm{new}}=\{(v_{2r+1},v_{2r+2}):0\le r\le p-1\}\). Start from the \(Z_{v_1}\) component of the local edge observable \(C_{v_0v_1}=\tfrac12(\mathbb{I}-Z_{v_0}Z_{v_1})\).
For a transport gate on \((a,b)\),
\[
e^{i\beta X_aX_b} Z_a e^{-i\beta X_aX_b}
=
\cos(2\beta)\, Z_a
+
\sin(2\beta)\, Y_aX_b
\]
and for a cost gate on \((a,b)\),
\[
e^{i\gamma C_{ab}} X_a e^{-i\gamma C_{ab}}
=
\cos(\gamma)\, X_a
+
\sin(\gamma)\, Y_aZ_b
\]
Thus the first transport gate on \((v_1,v_2)\) can produce a Pauli-string component supported on \(v_2\), and the subsequent cost gate on \((v_2,v_3)\) can produce a component supported on \(v_3\). Iterating this alternating mechanism through the chain shows that after \(p\) layers there is a Pauli-string component supported on \(v_{2p+1}\) whose coefficient is a product of trigonometric factors. Choosing, for example, all relevant transport angles equal to \(\pi/4\) and all relevant cost angles equal to \(\pi/2\), this coefficient is nonzero. Therefore the evolved observable has support reaching \(v_{2p+1}\).

Finally, \(\mathrm{dist}_{G_{\mathrm{int}}}(\{v_0,v_1\},v_{2p+1})=2p\), so the radius-\(2p\) upper bound is attained in general and cannot be improved.

\end{proof}

\begin{corollary}[Connected-correlation obstruction]
Let \(\rho=\bigotimes_{v\in V}\rho_v\) be a product state on the qubits, and let \(O_S\) and \(O_T\) be observables initially supported on disjoint subsets \(S,T\subseteq V\). If \(\mathrm{dist}_{G_{\mathrm{int}}}(S,T)>4p\), then the connected correlator of the Heisenberg-evolved observables vanishes:
\[
\operatorname{Tr}\!\bigl(\rho\, O_S^{(p)}O_T^{(p)}\bigr)
-
\operatorname{Tr}\!\bigl(\rho\, O_S^{(p)}\bigr)
\operatorname{Tr}\!\bigl(\rho\, O_T^{(p)}\bigr)
=
0
\]

\end{corollary}

\begin{proof}
By the operator-growth theorem, \(\operatorname{supp}(O_S^{(p)})\subseteq N_{2p}^{\mathrm{int}}(S)\) and \(\operatorname{supp}(O_T^{(p)})\subseteq N_{2p}^{\mathrm{int}}(T)\).
If \(\mathrm{dist}_{G_{\mathrm{int}}}(S,T)>4p,\)
then the two neighborhoods \(N_{2p}^{\mathrm{int}}(S)\) and \(N_{2p}^{\mathrm{int}}(T)\) are disjoint. Hence \(O_S^{(p)}\) and \(O_T^{(p)}\) act on disjoint tensor factors. Since \(\rho\) is a product state across qubits, expectations factorize across these disjoint supports, proving the vanishing of the connected correlator.

\end{proof}

\section{Additional Properties of Graph Augmentation}

\begin{proposition}[Tradeoff monotonicity of optimal augmentation size]
If
\(
1\le D_1\le D_2,
\)
then
\[
\mathcal A_{D_1}(G)\subseteq \mathcal A_{D_2}(G),
\qquad
\alpha_{D_1}(G)\ge \alpha_{D_2}(G)
\]
In particular,
\[
\operatorname{diam}(G)\le D
\implies
\alpha_D(G)=0
\]

\end{proposition}

\begin{proof}
If an augmentation \(F\) satisfies
\[
\operatorname{diam}(V,E\cup F)\le D_1
\]
then automatically
\(\operatorname{diam}(V,E\cup F)\le D_2\), so \(F\in\mathcal A_{D_2}(G)\). This proves \(\mathcal A_{D_1}(G)\subseteq \mathcal A_{D_2}(G)\). Taking minima of \(|F|\) over nested feasible sets gives \(\alpha_{D_1}(G)\ge \alpha_{D_2}(G)\). If the original graph already satisfies \(\operatorname{diam}(G)\le D\), then the empty augmentation \(F=\emptyset\) lies in \(\mathcal A_D(G)\), so \(\alpha_D(G)=0\).

\end{proof}

\begin{remark}[Complexity and approximability status]
The exact optimization problem defining \(\alpha_D(G)\) is computationally hard in general. Li, McCormick, and Simchi-Levi proved that the minimum-cardinality bounded-diameter edge-addition problem is NP-hard even when the target diameter is \(D=2\)~\cite{li1992}. Gao, Hare, and Nastos further showed that the parameterized decision version of diameter-\(t\) augmentation is \(W[2]\)-hard with respect to the number of added edges, for every fixed \(t\)~\cite{gao2013}.

On the positive side, Bil\`o, Gual\`a, and Proietti showed that the minimum-cardinality bounded-diameter edge-addition problem admits a polynomial-time \(O(\log n)\)-approximation algorithm and that this logarithmic dependence is asymptotically tight up to constant factors unless \(P=NP\)~\cite{bilo2012}. They also established a bicriteria guarantee: in polynomial time, one can add at most twice the optimal number of edges while obtaining diameter at most twice the target diameter~\cite{bilo2012}. Thus, although computing an exact optimum is unlikely to be efficient on arbitrary graphs, generic constructive heuristics and logarithmic-factor approximation algorithms are available.

\end{remark}

\section{Proof of Root-Centered Augmentation}

\begin{proof}
If
\(
\mathrm{dist}_G(r,v)\le p,
\)
then no new edge is added incident to \(v\), and
\(
\mathrm{dist}_{G_r^{(p)}}(r,v)\le p.
\)
If instead
\(
\mathrm{dist}_G(r,v)>p,
\)
then by construction the edge \(\{r,v\}\) is added, so
\(\mathrm{dist}_{G_r^{(p)}}(r,v)=1\le p\). Thus every vertex of \(G_r^{(p)}\) lies within distance at most \(p\) of the root \(r\). The root criterion from the previous subsection therefore implies \(\operatorname{diam}(G_r^{(p)})\le 2p\).

The number of added edges is exactly the number of vertices at distance greater than \(p\) from \(r\), namely \(|F_r^{(p)}|=n-|N_p^G(r)|\).
Minimizing this quantity over all roots is equivalent to maximizing \(|N_p^G(r)|\), which gives the stated bound on \(\alpha_{2p}(G)\).

\end{proof}

\section{Additional Properties of Graph Augmentation and Symmetry}
\label{app:graph-augmentation}

\begin{proof}
Because \(F_p^\star\in\mathcal A_{2p}(G)\), the interaction graph \(G_p^\star\) satisfies \(\operatorname{diam}(G_p^\star)\le 2p\).
Applying the full-lightcone criterion to this interaction graph shows that for every nonempty \(S\subseteq V\), the radius-\(2p\) support bound is the whole vertex set \(V\). This proves part (1).

Now let \(F\subseteq \binom{V}{2}\setminus E\) satisfy \(|F|<\alpha_{2p}(G)\). By optimality of \(F_p^\star\), such an \(F\) cannot lie in \(\mathcal A_{2p}(G)\). Therefore \(\operatorname{diam}(G_F)>2p\). The necessary-depth corollary applied to the interaction graph \(G_F\) then yields vertices \(u,v\in V\) for which \([U_p^\dagger O_u U_p,\; O_v]=0\) for all single-site operators \(O_u\) and \(O_v\) at depth \(p\). This proves part (2).

\end{proof}

\begin{remark}[Graph-agnostic character of the augmentation problem]
The definition of \(\alpha_D(G)\) depends only on the input graph \(G\) and the target diameter \(D\). It does not assume symmetry, regularity, expansion, bipartite structure, or any other special graph family. The optimization problem is therefore a generic design problem for the mixer transport geometry, while the underlying cost Hamiltonian on \(E\) remains unchanged.

\end{remark}

\begin{proposition}[Polynomial-time root-centered augmentation]
Fix a depth \(p\) and a root vertex \(r\in V\). Define
\[
F_r^{(p)}
:=
\{
\{r,v\}\in \binom{V}{2}\setminus E:
\mathrm{dist}_G(r,v)>p
\}
\]
Then the augmented interaction graph
\[
G_r^{(p)}:=(V,E\cup F_r^{(p)})
\]
satisfies
\(\operatorname{diam}(G_r^{(p)})\le 2p\). Moreover, \(|F_r^{(p)}|=n-|N_p^G(r)|\), where \(N_p^G(r):=\{v\in V:\mathrm{dist}_G(r,v)\le p\}\). Consequently, \(\alpha_{2p}(G)\le n-\max_{r\in V}|N_p^G(r)|\). Choosing a root \(r^\star\) maximizing \(|N_p^G(r)|\) yields a graph-agnostic polynomial-time algorithm for producing a feasible diameter-\(2p\) augmentation of this size.

\end{proposition}

\paragraph{Choosing \(E_{\mathrm{new}}\) to reduce diameter.}
Solving the exact minimization problem defining \(\alpha_D(G)\) may itself be combinatorial. For practical constructions, a simple greedy augmentation is often sufficient.

Fix a target diameter \(D\le 2p\). Initialize \(G_{\mathrm{int}}\leftarrow(V,E)\) and \(E_{\mathrm{new}}\leftarrow\emptyset\). Repeat:
\begin{enumerate}
    \item Compute a pair \((u,v)\) with maximum shortest-path distance in the current \(G_{\mathrm{int}}\).
    \item If \(\mathrm{dist}_{G_{\mathrm{int}}}(u,v)\le D\), stop.
    \item Otherwise add the shortcut edge \((u,v)\) to \(E_{\mathrm{new}}\) and update \(G_{\mathrm{int}}\).
\end{enumerate}

This explicitly trades off \(|E_{\mathrm{new}}|\) versus the achieved diameter and tends to create hub-like edges that shrink long geodesics rapidly.

A simple sufficient condition for \(\operatorname{diam}(G_{\mathrm{int}})\le 2p\) is the existence of a root vertex \(v_r\) such that every vertex satisfies \(\mathrm{dist}_{G_{\mathrm{int}}}(v,v_r)\le p\). Indeed, for any \(u,v\in V\), \(\mathrm{dist}_{G_{\mathrm{int}}}(u,v)\le \mathrm{dist}_{G_{\mathrm{int}}}(u,v_r)+\mathrm{dist}_{G_{\mathrm{int}}}(v_r,v)\le 2p\).

\paragraph{Implementing \(U_B\) (circuit synthesis).}

Each term in
\(
H_B(\beta)
\)
is a one- or two-qubit Pauli rotation.

For an added edge
\(
(i,j),
\)
we may implement
\begin{equation}
    e^{-i\beta_{ij}X_iX_j}
    =
    (H_i\otimes H_j)
    e^{-i\beta_{ij}Z_iZ_j}
    (H_i\otimes H_j)
\end{equation}
and a standard decomposition of
\(
e^{-i\beta Z_iZ_j}
\)
uses two CNOTs and an \(R_z(2\beta)\) rotation.

Since all \(X_iX_j\) terms commute with one another and with all \(X_k\) terms, the mixer layer can be parallelized according to an edge-coloring of \(G_M\).

\paragraph{Locality consequence at depth \(p\).}

By the operator-growth theorem, each Heisenberg-evolved edge term
\(
U^\dagger C_{jk} U
\)
is supported within
\(
N_{2p}^{\mathrm{int}}(\{j,k\}).
\)
Hence if the interaction graph is engineered so that
\(
\mathrm{diam}(G_{\mathrm{int}})\le 2p,
\)
then for every
\(
(j,k)\in E
\)
we have
\(
N_{2p}^{\mathrm{int}}(\{j,k\})=V,
\)
because for any vertex \(v\in V\),
\[
\mathrm{dist}_{G_{\mathrm{int}}}(v,\{j,k\})
\le
\mathrm{dist}_{G_{\mathrm{int}}}(v,j)
\le
\operatorname{diam}(G_{\mathrm{int}})
\le
2p
\]
Thus graph-distance support restrictions no longer force
\(
U^\dagger C_{jk} U
\)
to remain localized near the original edge.

This does not by itself guarantee a better approximation ratio for every parameter choice, nor does it rule out other bottlenecks such as symmetry constraints or concentration effects, but it does remove the specific graph-distance support restriction underlying fixed-depth locality arguments.

\paragraph{Hardware implications.}

If the device natively supports XX interactions, the added edges in
\(
E_{\mathrm{new}}
\)
correspond directly to additional couplers.
If only nearest-neighbor connectivity is available, one may either:
\begin{enumerate}
    \item Realize effective long-range XX interactions via SWAP networks, increasing depth, or
    \item Restrict \(E_{\mathrm{new}}\) to hardware edges while still choosing diameter-reducing augmentations of the hardware graph.
\end{enumerate}

The practical goal is to maximize effective mixer diameter reduction per unit two-qubit-gate cost.

\subsection{Symmetry Sectors from Graph Automorphisms}

Let \(\Gamma\) be a subgroup of permutations of \(V\). For each \(g\in\Gamma\), let \(P_g\) denote the induced qubit-permutation unitary.

\begin{theorem}[Automorphism-sector invariance]
Assume that for every \(g\in\Gamma\),
\[
g(E)=E,
\qquad
g(E_{\mathrm{new}})=E_{\mathrm{new}}
\]
and that in each layer \(\ell\), the mixer parameters are constant on \(\Gamma\)-orbits of transport edges:
\[
\beta_{g(i)g(j)}^{(\ell)}=\beta_{ij}^{(\ell)}
\qquad
\text{for all }(i,j)\in E_{\mathrm{new}}
\]
Then
\[
[P_g,C]=0,
\qquad
[P_g,U_B(\beta^{(\ell)})]=0
\]
for every \(g\in\Gamma\) and every layer \(\ell\). Hence the full depth-\(p\) unitary commutes with the \(\Gamma\)-action:
\[
[P_g,U_p]=0
\qquad
\text{for all }g\in\Gamma
\]
Consequently every symmetry sector of the permutation representation is invariant under the transport-augmented QAOA dynamics. In particular, if the initial state is \(\Gamma\)-invariant, then the evolution remains inside the fixed-point subspace
\[
\mathcal H^\Gamma
:=
\{
\ket{\psi}:
P_g\ket{\psi}=\ket{\psi}
\text{ for all }g\in\Gamma
\}
\]

\end{theorem}

\begin{proof}
For the MaxCut cost Hamiltonian,
\[
C=\sum_{(i,j)\in E}\frac12(\mathbb{I}-Z_iZ_j)
\]
and the assumption \(g(E)=E\) implies that \(P_g\) permutes the edge terms among themselves. Therefore
\[
P_g C P_g^\dagger=C
\]
which is equivalent to \([P_g,C]=0\).

Similarly,
\[
H_B(\beta^{(\ell)})
=
\beta_0^{(\ell)}\sum_{u\in V}X_u
+
\sum_{(i,j)\in E_{\mathrm{new}}}\beta_{ij}^{(\ell)}X_iX_j
\]
The single-qubit term is permutation-invariant, and the assumptions \(g(E_{\mathrm{new}})=E_{\mathrm{new}}\) together with
\[
\beta_{g(i)g(j)}^{(\ell)}=\beta_{ij}^{(\ell)}
\]
show that \(P_g\) also permutes the two-qubit mixer terms without changing their coefficients. Hence
\[
P_g H_B(\beta^{(\ell)}) P_g^\dagger=H_B(\beta^{(\ell)})
\]
so \([P_g,H_B(\beta^{(\ell)})]=0\), and therefore
\[
[P_g,U_B(\beta^{(\ell)})]=0
\]

Since \(U_p\) is a product of layers \(U_B(\beta^{(\ell)})U(C,\gamma_\ell)\), it follows that \([P_g,U_p]=0\) for every \(g\in\Gamma\). Any symmetry sector of the \(\Gamma\)-representation is therefore invariant under \(U_p\). In particular, if a state satisfies \(P_g\ket{\psi}=\ket{\psi}\) for every \(g\in\Gamma\), then so does \(U_p\ket{\psi}\), so the fixed-point subspace \(\mathcal H^\Gamma\) is invariant.

\end{proof}

\end{document}